\documentclass{article}

\usepackage{amsmath,amsfonts}

\usepackage{amsfonts,amssymb,amsthm}
\usepackage{graphicx}
\usepackage{enumerate}
\usepackage{authblk}

\usepackage{hyperref} 

\graphicspath{{./figures/}}

\usepackage[color=red!40]{todonotes}
\usepackage[linesnumbered,ruled,vlined]{algorithm2e}
\usepackage{microtype}

\usepackage{thmtools}
\usepackage{thm-restate}

\usepackage{xcolor,mathdots}
\definecolor{color1}{RGB}{228,26,28}
\definecolor{color2}{RGB}{55,126,184}
\definecolor{color3}{RGB}{77,175,74}
\definecolor{color4}{RGB}{152,78,163}
\definecolor{color5}{RGB}{255,127,0}
\definecolor{color6}{RGB}{255,255,51}
\definecolor{color7}{RGB}{166,86,40}
\definecolor{color8}{RGB}{247,129,191}
\usepackage{tikz-cd}
\usetikzlibrary{arrows,shapes.geometric,positioning,trees,calc,fit}
\tikzstyle{tnode} = [circle,inner sep=0pt,minimum size=12pt]
\tikzstyle{leaf} = [draw,circle,inner sep=0pt,minimum size=12pt,line width=0.5pt]
\tikzstyle{inner} = [draw,circle,inner sep=0pt,minimum size=5pt]
\tikzstyle{tclade} = [draw,dashed,regular polygon,regular polygon sides=3,inner sep=0pt,minimum size=36pt,yshift=-0.2cm]
\tikzstyle{lowlight} = [fill=blue!20] 
\tikzstyle{highlight} = [fill=red!20,tnode]
\tikzstyle{lowlightedge} = [preaction={color=blue!20,line width=3pt}]
\tikzstyle{highlightedge} = [preaction={draw,color=red!20,line width=3pt}]
\tikzstyle{highlightedgecolor} = [preaction={color=#1,draw,line width=3pt}]
\newcommand{\tikzmark}[1]{\tikz[overlay,remember picture] \node (#1) {};}
\newcommand{\DrawBox}[3][]{%
    \tikz[overlay,remember picture]{
    \fill[fill opacity=0.2, #1]
      ($({#2})+(-0.2em,0.9em)$) rectangle
      ($({#3})+(0.2em,-0.3em)$);}
}

\newtheorem{thm}{Theorem}

\newtheorem{prop}[thm]{Proposition}
\newtheorem{conj}[thm]{Conjecture}

\theoremstyle{definition}

\newtheorem{dfn}[thm]{Definition}

\newcommand{\ZZ}{\mathbb Z}
\newcommand{\EE}{\mathbb E}

\DeclareMathOperator{\OLA}{OLA}
\newcommand{\RBT}{\mathcal{RBT}} 
\newcommand{\RBThat}{\widehat{\mathcal{RBT}}}
\newcommand{\codes}{\mathcal{C}}
\DeclareMathOperator{\labelfn}{{\tt label}}

\DeclareMathOperator{\theight}{ht}
\newcommand{\CF}[1]{{#1}.{\tt CF}}
\newcommand{\CS}[1]{{#1}.{\tt CS}}

\title{Vector encoding of phylogenetic trees by ordered leaf attachment}

\author[1,2]{Harry Richman}
\author[3]{Cheng Zhang}
\author[1,4,5,6]{Frederick A. Matsen IV}

\affil[1]{{\small Computational Biology Program, Fred Hutchinson Cancer Center, Seattle, USA}}
\affil[2]{{\small Mathematics Division, National Center for Theoretical Sciences, Taipei, Taiwan}}
\affil[3]{{\small School of Mathematical Sciences and Center for Statistical Science, Peking University, Beijing, China}}
\affil[4]{{\small Department of Statistics, University of Washington, Seattle, USA}}
\affil[5]{{\small Department of Genome Sciences, University of Washington, Seattle, USA}}
\affil[6]{{\small Howard Hughes Medical Institute, Fred Hutchinson Cancer Center, Seattle, USA}}

\begin{document}

\maketitle

\begin{abstract}
  As part of work to connect phylogenetics with machine learning, there has been considerable recent interest in vector encodings of phylogenetic trees.
  We present a simple new ``ordered leaf attachment'' (OLA) method for uniquely encoding a binary, rooted phylogenetic tree topology as an integer vector.
  OLA encoding and decoding take linear time in the number of leaf nodes, and the set of vectors corresponding to trees is a simply-described subset of integer sequences.
  The OLA encoding is unique compared to other existing encodings in having these properties.
  The integer vector encoding induces a distance on the set of trees, and we investigate this distance in relation to the NNI and SPR distances.
\end{abstract}


\section{Introduction}

Phylogenetics is the study of evolutionary histories, and how to infer such histories from observed data.
An evolutionary history is commonly formalized as the data of a leaf-labeled binary rooted tree (a ``tree topology'').
Inference and analysis of phylogenetic trees is typically done using a tree as a discrete structure.

However, recent work in phylogenetics has begun to consider the use of vector representations of trees.
This is motivated by recent advances in machine learning, where vector representations are the standard output of many classes of models.
It is also motivated by the desire to represent collections of trees.

The first of the more recent contributions to this literature includes work by Penn et al.~\cite{penn-etal} introducing the Phylo2Vec encoding, which is an integer vector representation of a tree topology.
This system encodes a phylogenetic tree with $n$ leaves as a list of $n - 1$ integers $(a_1, \ldots, a_{n-1})$, where the $i$-th entry satisfies $0 \leq a_i \leq 2 i - 2$.
However, this encoding requires quadratic time to compute.
The HOP metric paper~\cite{chauve-etal} introduced a linear-time encoding, but both the encoding itself and the resulting set of vectors in this encoding are rather complicated.
Both of these papers investigated the connection between encoding distances and traditional tree metrics.

In this paper, we introduce a simple new method for encoding a tree topology as a list of integers that can be computed in linear time.
We call this the ``ordered leaf attachment'' (OLA) code.
The OLA code is a bijection between the set of rooted tree topologies on $n$ leaves and a convex set of length $n-1$ vectors defined by a simple set of coordinate-wise inequalities.

The OLA encoding induces a distance on rooted tip-ordered trees.
We show that this OLA distance is bounded ``on average'' by a constant factor of the NNI distance, in a sense that we make precise.
We conjecture, and provide empirical evidence for, an asymptotic upper bound for OLA distance by SPR distance in terms of (average) tree height.
A construction based on the OLA code has already proven itself useful in work defining a variational autoencoder on phylogenetic trees~\cite{xie2025phylovae}.

\section{Results}

\subsection{Informal description of OLA encoding}

Here we give an informal description of the OLA code, which will be sufficient to give intuition and present results. 
A detailed description and further discussion is given in Section~\ref{sec:methods}.

The basic idea behind the OLA encoding is that a single tree, with a total ordering on its leaves, can be thought of as a sequence of trees obtained by adding the leaves one-by-one.
In the OLA encoding, the $i$-th entry of the vector records the location where the $i$-th leaf was added to the tree.

We first consider a special case: say a phylogenetic tree is \emph{Yule-type} if, in the process of building up the tree, each new leaf is added as a sister to a previously-added leaf.
(This terminology comes from the Yule random tree model~\cite[Chapter 2.5]{semple-steel}.)
For a Yule-type tree, the $i$-th entry of the OLA encoding is just the label of the sister leaf at the step when the $i$-th leaf is added.
For example, the tree shown in Figure~\ref{fig:yule-ola-code} has OLA code \texttt{(0,0,2,1)}.

\begin{figure}[h]
  \centering
  \raisebox{-0.5\height}{\begin{tikzpicture}[scale=0.5,
    level 1/.style={level distance=1cm},
  ]
    \node {$\rho$} [grow=down]
      child {node[lowlight,inner] (P) {}
        child {node[leaf,highlight] {0}}
        child {node[leaf,lowlight] (Q) {1}}
      };
    \draw[lowlightedge] (P) -- (Q);
    \draw (P) -- (Q);
  \end{tikzpicture}}
  $\quad\rightsquigarrow\quad$  
  \raisebox{-0.5\height}{\begin{tikzpicture}[scale=0.5,
    level 1/.style={level distance=1cm},
  ]
  \node {$\rho$} [grow=down]
      child {node[inner,label=above right:{}] {}
        child {node[inner,label={},lowlight] {}
          child {node[leaf,highlight] {0}}
          child {node[leaf,lowlight] {2}}
        }
        child {node[leaf] {1}}
      };
  \end{tikzpicture}}
  \quad$\rightsquigarrow$\quad
  \raisebox{-0.5\height}{\begin{tikzpicture}[scale=0.5,
    level 1/.style={level distance=1cm},
  ]
  \node {$\rho$} [grow=down]
      child {node[inner,label=above right:{}] {}
        child {node[inner,label={},left=0.2cm] {}
          child {node[leaf] {0}}
          child {node[inner,label={},lowlight] (P) {}
            child {node[leaf,highlight] {2}}
            child {node[leaf,lowlight] (Q) {3}}
          }
        }
        child {node[leaf,right=0.2cm] {1}}
      };
    \draw[lowlightedge] (P) -- (Q);
    \draw (P) -- (Q);
  \end{tikzpicture}}
  \quad$\rightsquigarrow$\quad
  \raisebox{-0.5\height}{\begin{tikzpicture}[scale=0.5,
    level 1/.style={level distance=1cm},
  ]
  \node {$\rho$} [grow=down]
      child {node[inner,label=above right:{}] {}
        child {node[inner,label={},left=0.5cm] {}
          child {node[leaf] {0}}
          child {node[inner,label={}] {}
            child {node[leaf] {2}}
            child {node[leaf] {3}}
          }
        }
        child {node[inner,label={},lowlight,right=0.5cm] (P) {}
          child {node[leaf,highlight] {1}}
          child {node[leaf,lowlight] (Q) {4}}
        }
      };
    \draw[lowlightedge] (P) -- (Q);
    \draw (P) -- (Q);
  \end{tikzpicture}}
\caption{%
  OLA-encoding a Yule-type tree: \texttt{(0,0,2,1)}.
  Yule-type trees are those where each new leaf is added as a sister to a previously-added leaf, and for these trees the OLA encoding simply records the label of the sister leaf.
  At each step, the new leaf is highlighted in blue, and the sister leaf is highlighted in red.
}
\label{fig:yule-ola-code}
\end{figure}
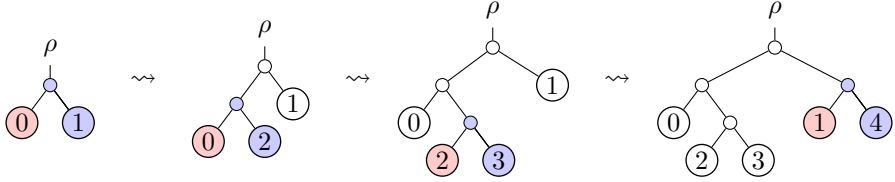

In the general case, when a new leaf is added, it may be added as the sister of an internal node.
Here a complication arises, since the internal nodes do not come equipped with labels.
However, using this idea to ``record each sister-node label,'' the problem of encoding a phylogenetic tree can be reduced to the problem of defining a consistent method for labeling internal nodes of a phylogenetic tree.

Our proposed OLA encoding provides one such labeling of internal nodes.
For a tree with $n$ leaves, the internal nodes are labeled with negative integers $-1, -2, \ldots, -n + 1$.
Each time a new leaf is added to a tree, a new internal node is created.
When we add a leaf labeled by $i$, we then label the new internal node $-i$.
This gives a well-defined labeling of internal nodes, which we consider ``canonical.''
An example of these internal node labels is shown in Figure~\ref{fig:non-yule-ola-code}.
For a larger example, see Appendix~\ref{sec:ola-ex-large}.
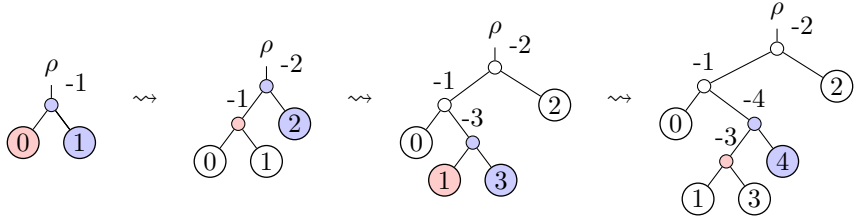
\begin{figure}[h]
  \centering
  \raisebox{-0.5\height}{\begin{tikzpicture}[scale=0.5,
    level 1/.style={level distance=1cm},
  ]
    \node {$\rho$} [grow=down]
      child {node[lowlight,inner,label=above right:{-1}] (P) {}
        child {node[leaf,highlight] {0}}
        child {node[leaf,lowlight] (Q) {1}}
      };
    \draw[lowlightedge] (P) -- (Q);
    \draw (P) -- (Q);
  \end{tikzpicture}}
  $\quad\rightsquigarrow\quad$  
  \raisebox{-0.5\height}{\begin{tikzpicture}[scale=0.5,
    level 1/.style={level distance=1cm},
  ]
  \node {$\rho$} [grow=down]
      child {node[inner,lowlight,label=above right:{-2}] {}
        child {node[highlight,inner,label={-1}] {}
          child {node[leaf] {0}}
          child {node[leaf] {1}}
        }
        child {node[leaf,lowlight] {2}}
      };
  \end{tikzpicture}}
  \quad$\rightsquigarrow$\quad
  \raisebox{-0.5\height}{\begin{tikzpicture}[scale=0.5,
    level 1/.style={level distance=1cm},
  ]
  \node {$\rho$} [grow=down]
      child {node[inner,label=above right:{-2}] {}
        child {node[inner,label={-1},left=0.2cm] {}
          child {node[leaf] {0}}
          child {node[inner,label={-3},lowlight] (P) {}
            child {node[leaf,highlight] {1}}
            child {node[leaf,lowlight] (Q) {3}}
          }
        }
        child {node[leaf,right=0.2cm] {2}}
      };
  \end{tikzpicture}}
  \quad$\rightsquigarrow$\quad
  \raisebox{-0.5\height}{\begin{tikzpicture}[scale=0.5,
    level 1/.style={level distance=1cm},
  ]
  \node {$\rho$} [grow=down]
      child {node[inner,label=above right:{-2}] {}
        child {node[inner,label={-1},left=0.5cm] {}
          child {node[leaf] {0}}
          child {node[lowlight,inner,label={-4},right=0.2cm] {}
            child {node[highlight,inner,label={-3}] {}
              child {node[leaf] {1}}
              child {node[leaf] {3}}
            }
            child {node[lowlight,leaf] {4}}
          }
        }
        child {node[leaf,right=0.2cm] {2}}
      };
  \end{tikzpicture}}
\caption{%
OLA-encoding a non-Yule-type tree: \texttt{(0,-1,1,-3)}.
For non-Yule-type trees, we introduce a canonical labeling of internal nodes, where each internal node is labeled by $-i$ when the $i$-th leaf is added.
Using this convention, we construct the OLA encoding by adding the label of the sister node at each step, whether or not that sister is a leaf.
At each step, the new leaf is highlighted blue, and the sister node is highlighted red.
}
\label{fig:non-yule-ola-code}
\end{figure}

\subsection{Statement of bijection}
Now we can state the main property of the ordered leaf attachment (OLA) code for a leaf-ordered binary, rooted tree.
Let $\RBThat_n$ denote the set of rooted, binary trees with $n$ leaves equipped with a linear order on the leaf set.
The OLA encoding gives a bijection between $\RBThat_n$ and the set of valid integer vectors
\begin{equation}
  \codes_{n - 1} = \{ (a_1, a_2, \ldots, a_{n - 1}) \in \ZZ^{n - 1} : -i < a_i < i \}.
\end{equation}

\begin{restatable}{thm}{OLAEncodingDecoding}
  For any $n \geq 2$, the OLA encoding and decoding algorithms (Algorithms~\ref{alg:ola-encode} and \ref{alg:ola-decode}) define a pair of inverse bijections
  \begin{equation}
    \Phi : \RBThat_n \to \codes_{n - 1} \qquad\text{and}\qquad \Psi : \codes_{n - 1} \to \RBThat_n.
  \end{equation}
\end{restatable}

The proof of this theorem is given in the Methods section (Section~\ref{sec:methods}).
The OLA encodings of all trees for $n = 4$ is shown in Appendix~\ref{sec:4-leaves-ex}.

\subsection{OLA distance and standard phylogenetic distances}\label{sec:ola-distance-compare}

The OLA encoding is a bijection between the set of rooted, binary trees with $n$ leaves and the set of integer vectors $\codes_{n - 1}$.
This itself is useful. 
However, one may wish to understand to what extent similar OLA vectors correspond to similar trees. 
We introduce a tree-distance based on OLA vectors, and show that it is related to the NNI and SPR distances.

\begin{dfn}
  The {\em OLA distance} between two leaf-ordered binary rooted trees $T, T'$ on $n$ leaves is the Hamming distance between their OLA encodings.
\end{dfn}

\subsubsection{Comparison to NNI distance}

We would ideally prefer if the OLA distance were close to a standard phylogenetic distance, for example within a constant factor.
Unfortunately, this is not the case when compared with NNI distance.

In fact, in the worst case a single NNI move can produce a tree with arbitrarily large OLA distance from some starting tree.
See Figure~\ref{fig:nni-large-ola-dist} for such an example.
This example shows two trees on $n$ leaves which are related by an NNI move, modifying the edge highlighted in red.
Their OLA codes differ in $n - 2$ entries.
This is the largest possible OLA distance between two trees on $n$ leaves (since the first entry of the OLA code is always $0$).

\begin{figure}[h]
  \begin{minipage}{0.45\textwidth}
    \centering
    \begin{tikzpicture}[scale=0.6,
      level 1/.style={level distance=1cm},
    ]
    \node {$\rho$} [grow=down,]
      child {node[inner,label=above right:{-3}] {}
        child {node[inner,label={-4}] {}
          child {node[left=-2.0mm] {$\iddots$}
            child {node[left=4mm,inner,label=above left:{-$n$+1}] {}
              child {node[highlight,inner,label={-1}] (-1) {}
                child {node[highlight,inner,label={-2}] (-2) {}
                  child {node[leaf] {0}}
                  child {node[leaf] {2}}
                }
                child {node[leaf] {1}}
              }
              child {node[leaf] {\small $n$-1}}
            }
          }
          child {node[leaf] {4}}
        }
        child {node[leaf] {3}}
      };
    \draw[highlightedge] (-1) -- (-2);
    \draw (-1) -- (-2);
    \node[fit=(current bounding box),inner ysep=0.2cm,inner xsep=0] {};
    \end{tikzpicture}
    \\
    {\tt (0,0,-1,-1,...,-1)}
  \end{minipage}
  \begin{minipage}{0.45\textwidth}
    \centering
    \begin{tikzpicture}[scale=0.6,
      level 1/.style={level distance=1cm},
    ]
    \node {$\rho$} [grow=down,]
      child {node[inner,label=above right:{-3}] {}
        child {node[inner,label={-4}] {}
          child {node[left=-1.0mm] {$\iddots$}
            child {node[left=4mm,inner,label=above left:{-$n$+1}] {}
              child {node[highlight,inner,label={-2}] (-2) {}
                child {node[highlight,inner,label={-1}] (-1) {}
                  child {node[leaf] {0}}
                  child {node[leaf] {1}}
                }
                child {node[leaf] {2}}
              }
              child {node[leaf] {\small $n$-1}}
            }
          }
          child {node[leaf] {4}}
        }
        child {node[leaf] {3}}
      };
    \draw[highlightedge] (-1) -- (-2);
    \draw (-1) -- (-2);
    \node[fit=(current bounding box),inner ysep=0.2cm,inner xsep=0] {};
    \end{tikzpicture}
    \\
    {\tt (0,-1,-2,-2,...,-2)}
  \end{minipage}
  \caption{Two trees which differ by an NNI move, with large OLA distance.}
  \label{fig:nni-large-ola-dist}
\end{figure}
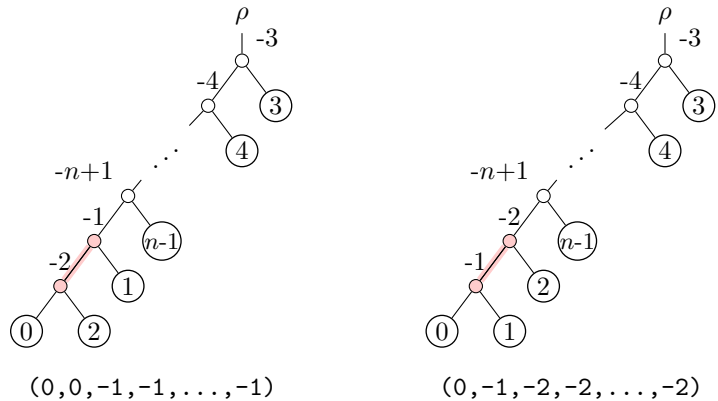

However, as we show next, the expected OLA distance between two trees which differ by an NNI move is constant in expectation, when averaged over all moves from a given starting tree.

\begin{thm}\label{thm:nni-move-bound}
Suppose $T$ is a fixed phylogenetic tree in $\RBThat_n$, and let $T'$ be a random tree which differs from $T$ by an NNI move, chosen uniformly.
Then the expectation of the OLA distance has the constant upper bound
\[
  \EE(d_{\OLA}(T, T')) \leq 4.
\]
\end{thm}

\begin{proof}[Proof of Theorem \ref{thm:nni-move-bound}]
  Consider trees $T$ and $T'$ which differ by an NNI move.

  Suppose the NNI move is achieved by first contracting an edge in $T$ which lies between nodes with canonical labels $-j$ and $-k$, with $-k$ on top (Figure~\ref{fig:nni-move}).
  After the NNI move, in $T'$, the modified edge either still has label $-k$ on top, or switches to have label $-j$ on top.
  In the first case, where the top-label stays the same, we claim that $d_{\OLA}(T, T')$ is exactly $1$.
  This is because the only change in the OLA code is at entry $a_j$.

  In the second case, where the top-label changes in $T'$, let $c_k$ denote the number of entries in the OLA encoding of $T$ which are $-k$'s.
  We claim that $d_{\OLA}(T, T')$ is bounded above by $1 + c_k$.
  Indeed, if some entry of the OLA encoding changes between $T$ and $T'$, then that entry is either $a_j$ or $a_k$ (which accounts for the ``1''),
  or else some entry $a_\ell$, 
  for a leaf which attaches just above the NNI-modified edge. 
  In the latter case, we would have $a_\ell = -k$ in the OLA encoding of $T$ and $a_\ell = -j$ in the encoding of $T'$.

  Thus, accounting for both cases,
  \[
    \EE(d_{\OLA}(T, T')) \leq \EE(1 + c_k)
  \]
  where $c_k = c_{k(T')}$ depends on the NNI move chosen (which determines $T'$).
  In a tree with $n$ leaves, there are $2n - 4$ possible NNI moves.
  For each index $k$ there are at most four NNI moves with ``upper node'' $-k$.
  In other words, for each fixed $k$ there are at most four NNI-neighbors $T'$ such that $k(T') = k$.
  Thus,
  \begin{align*}
    \EE(1 + c_{k(T')}) &= 1 + \frac{1}{2n - 4} \sum_{T'} c_{k(T')} \\
    &\leq 1 + \frac{4}{2n - 4} \sum_{k = 1}^n c_k \\
    &= 1 + \frac{4n - 4}{2n - 4}
  \end{align*}
  where the first sum varies over $T'$ which are one NNI move away from $T$.
  In the last line, we use the fact that $\sum_{k = 1}^n c_k \leq n - 1$, since the OLA encoding of $T$ has $n - 1$ entries.
  Finally, we have that
  \[
    1 + \frac{4n - 4}{2n - 4} 
    \leq 4 \qquad\text{when } n \geq 4,
  \]
  so the claim is proved.
\end{proof}

  \begin{figure}[h]
    \begin{minipage}{0.95\textwidth}
      \centering
      \raisebox{-0.5\height}{\begin{tikzpicture}[scale=0.5,
        level 1/.style={level distance=1cm},
      ]
      \node {} [grow=down]
          child {node[inner,label=above right:{-$k$}] {}
            child {node[inner,label={-$j$},left=0.5cm] {}
              child {node[left=0.2cm] {}
                child {node[tclade] {$i$}}
              }
              child {node[right=0.2cm] {}
                child {node[tclade] {$j$}}
              }
            }
            child {node[right=0.3cm] {}
              child {node[tclade] {$k$}}
            }
          };
      \end{tikzpicture}}
      \quad$\rightsquigarrow$\quad
      \raisebox{-0.5\height}{\begin{tikzpicture}[scale=0.5,
        level 1/.style={level distance=1cm},
      ]
      \node {} [grow=down]
          child {node[inner,label=above right:{-$k$}] {}
            child {node[left=0.3cm] {}
              child {node[tclade] {$i$}}
            }
            child {node[inner,label={-$j$},right=0.5cm] {}
              child {node[left=0.2cm] {}
                child {node[tclade] {$k$}}
              }
              child {node[right=0.2cm] {}
                child {node[tclade] {$j$}}
              }
            }
          };
      \end{tikzpicture}}
    \end{minipage}\\[2em]
    \begin{minipage}{0.95\textwidth}
      \centering
      \raisebox{-0.5\height}{\begin{tikzpicture}[scale=0.5,
        level 1/.style={level distance=1cm},
      ]
      \node {} [grow=down]
          child {node[inner,label=above right:{-$k$}] {}
            child {node[inner,label={-$j$},left=0.5cm] {}
              child {node[left=0.2cm] {}
                child {node[tclade] {$i$}}
              }
              child {node[right=0.2cm] {}
                child {node[tclade] {$j$}}
              }
            }
            child {node[right=0.3cm] {}
              child {node[tclade] {$k$}}
            }
          };
      \end{tikzpicture}}
      \quad$\rightsquigarrow$\quad
      \raisebox{-0.5\height}{\begin{tikzpicture}[scale=0.5,
        level 1/.style={level distance=1cm},
      ]
      \node {} [grow=down]
          child {node[inner,label=above right:{-$j$}] {}
            child {node[inner,label={-$k$},left=0.5cm] {}
              child {node[left=0.2cm] {}
                child {node[tclade] {$i$}}
              }
              child {node[right=0.2cm] {}
                child {node[tclade] {$k$}}
              }
            }
            child {node[right=0.3cm] {}
              child {node[tclade] {$j$}}
            }
          };
      \end{tikzpicture}}
    \end{minipage}
    \caption{Trees differing by an NNI move. For the purposes of the proof of Theorem~\ref{thm:nni-move-bound}, the left tree is $T$ and the right tree is $T'$.}
    \label{fig:nni-move}
  \end{figure}
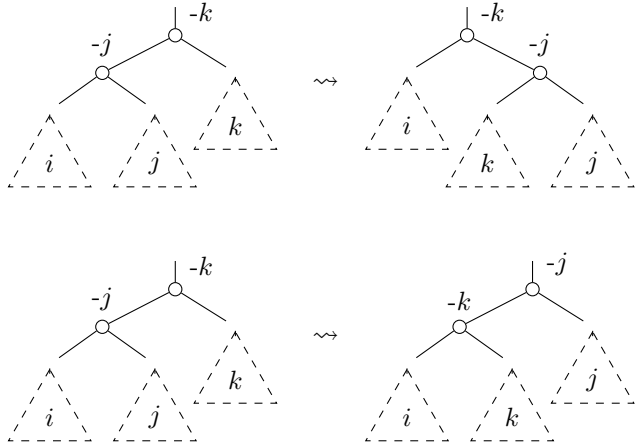
  
The bound in Theorem~\ref{thm:nni-move-bound} is in agreement with computational experiment, as shown in Figure~\ref{fig:nni-neighbors}.

\subsubsection{Comparison to SPR distance}

A common method for exploring the space of phylogenetic trees is through the use of subtree-prune-regraft (SPR) moves.
A (rooted) SPR move consists of choosing a subtree to prune, and then choosing a location on the remaining tree at which to regraft the pruned subtree.

\begin{figure}[h]
  \begin{minipage}{0.45\textwidth}
    \centering
    \begin{tikzpicture}[scale=0.6,
      level 1/.style={level distance=1cm},
    ]
    \node {$\rho$} [grow=down,]
      child {node[inner,label=above right:{-1}] {}
        child {node[inner,label={-2}] {}
          child {node[inner,label={-3}] {}
            child {node[left=-2.0mm] {$\iddots$}
              child {node[left=4mm,inner,label=above left:{-$n$+2}] (-n+2) {}
                child {node[inner,label=above left:{-$n$+1}] (-n+1) {}
                  child {node[leaf,highlight] (0) {0}}
                  child {node[leaf] {\small $n$-1}}
                }
                child {node[leaf] {\small $n$-2}}
              }
            }
            child {node[leaf] {3}}
          }
          child {node[leaf] {2}}
        }
        child {node[leaf] {1}}
      };
    \draw[highlightedge] (0) -- (-n+1);
    \draw (0) -- (-n+1);
    \node[fit=(current bounding box),inner ysep=0.2cm,inner xsep=0] {};
    \end{tikzpicture}
    \\
    {\tt (0,0,0,...,0)}
  \end{minipage}
  \begin{minipage}{0.45\textwidth}
    \centering
    \begin{tikzpicture}[scale=0.6,
      level 1/.style={level distance=1cm},
    ]
    \node {$\rho$} [grow=down,]
      child {node[inner,label=above right:{-1}] (-1) {}
        child {node[leaf,highlight] (0) {0}}
        child {node[inner,label={-2}] {}
          child {node[leaf] {1}}
          child {node[inner,label={-3}] {}
            child {node[leaf] {2}}
            child {node[inner,label={-4}] {}
              child {node[leaf] {3}}
              child {node[right=-1.0mm] {$\ddots$}
                child {node[right=4mm,inner,label=above right:{-$n$+1}] {}
                  child {node[leaf] {\small $n$-2}}
                  child {node[leaf] {\small $n$-1}}
                }
              }
            }
          }
        }
      };
    \draw[highlightedge] (0) -- (-1);
    \draw (0) -- (-1);
    \node[fit=(current bounding box),inner ysep=0.2cm,inner xsep=0] {};
    \end{tikzpicture}
    \\
    {\tt (0,1,2,...,n-2)}
  \end{minipage}
  \caption{Two trees which differ by an SPR move, with large OLA distance.}
  \label{fig:spr-large-ola-dist}
\end{figure}
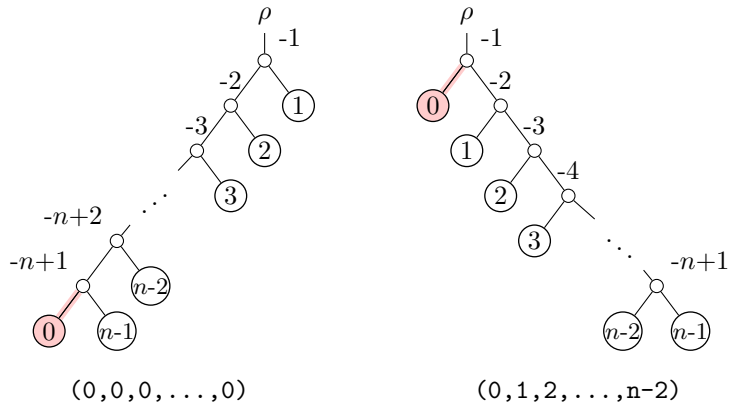

We introduce some standard terminology: the \emph{height} of a rooted tree is the number of edges between the root node and the farthest leaf node.
We use $\theight(T)$ to denote the height of the tree $T$.

\begin{dfn}
  The {\em height} of a (rooted) tree is the distance from the root to the farthest leaf.
\end{dfn}

\begin{conj}\label{conj:spr-move-bound}
  Suppose $T \in \RBThat_n$ is a fixed phylogenetic tree, and let $T'$ be a random tree which differs from $T$ by a single SPR move, chosen uniformly.
  Then the expectation of the OLA distance satisfies the asymptotic upper bound
  \[
    \EE(d_{\OLA}(T, T')) = O(\theight(T)),
  \]
  where $\theight(T)$ denotes the height of $T$.
\end{conj}
  
Here is a heuristic argument for why we believe Conjecture~\ref{conj:spr-move-bound} holds.
A chosen SPR move may be expressed as a combination of multiple NNI moves, where the number of NNI moves needed is equal to the distance between the prune-location and the regraft-location in the starting tree.
This distance is bounded by $2 \cdot \theight(T)$.
Each NNI move induces an expected OLA-distance of $O(1)$, by Theorem~\ref{thm:nni-move-bound}, so these distances add up to $O(\theight(T))$.
However, the place where this argument fails to be rigorous is that the distribution of the NNI moves which appear from ``decomposing'' a uniformly-random SPR move is different from a uniformly-random NNI move, so we cannot directly use Theorem~\ref{thm:nni-move-bound} in this manner.

Nevertheless, Conjecture~\ref{conj:spr-move-bound} is supported by the following computational experiment.
In Figure~\ref{fig:spr-neighbors}, we choose a random tree, sample some of its SPR neighbors uniformly, and then compute the average OLA distance to these neighbors.
This average OLA distance is plotted against the height of the starting tree.


\begin{figure}[h]
  \centering
  \includegraphics*[scale=0.8]{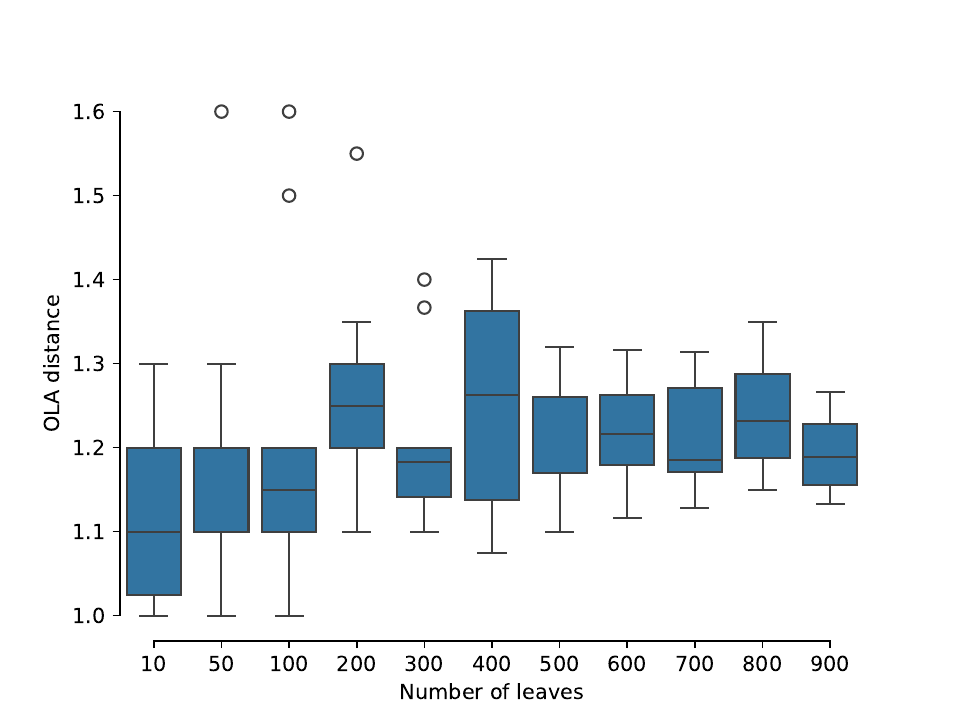}
  \caption{The expected OLA distance between a tree and its NNI neighbors. 
  The size of the tree varies from 10 leaves to 900 leaves.
  Theorem~\ref{thm:nni-move-bound} states that the expected OLA distance is at most $4$, while the simulated data suggests the actual expected distance is lower.}
  \label{fig:nni-neighbors}
\end{figure}

\begin{figure}[h]
  \centering
  \includegraphics*[scale=0.8]{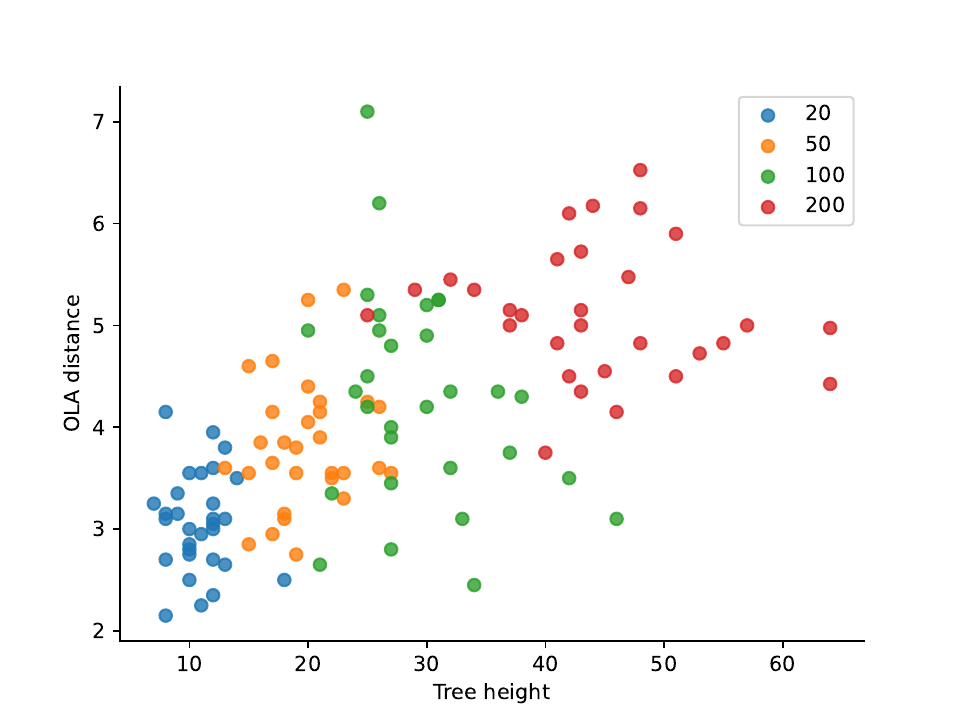}
  \caption{The expected OLA distance between a tree and its SPR neighbors. 
  The size of the tree varies from 20 leaves to 200 leaves, and is indicated by color.
  Conjecture~\ref{conj:spr-move-bound} and the corresponding heuristic argument suggest that this expected OLA distance is bounded by $8$ times the tree height, which is easily satisfied by the data in this simulation.}
  \label{fig:spr-neighbors}
\end{figure}

\section{Discussion}

In this paper, we introduce a new encoding of phylogenetic trees: the ``ordered leaf attachment'' (OLA) code.
This encoding is unique in being linear-time computable, and in having a simple description of the set of vectors corresponding to trees.
The OLA encoding induces a distance on the set of trees, and we showed that this distance is in expectation approximately constant between NNI neighbors.
We conjecture that the distance between a tree $T$ and an SPR neighbor $T'$ is in expectation asymptotically bounded above by the height of $T$.

There has been a long history of developing bijections between phylogenetic trees and other mathematical objects.
In 1918, Pr\"ufer~\cite{prufer} defined an encoding, now known as the ``Pr\"ufer code'', of trees which have an ordering on all nodes (both leaf nodes and internal nodes), which proves combinatorially that the number of trees on $n$ vertices is $n^{n - 2}$~\cite{cayley}.
Erd\H{o}s--Szekely~\cite{erdos-szekely} found a bijection between rooted phylogenetic trees and certain partitions of sets of integers.
Diaconis and Holmes~\cite{diaconis-holmes-matchings,diaconis-holmes-walks} found a bijection from rooted phylogenetic trees on $n$ leaves to perfect matchings on $2(n - 1)$ nodes.

There has been a recent surge of interest in developing embeddings of phylogenetic trees, initiated by the Phylo2Vec encoding introduced by Penn et al.~\cite{penn-etal}.
This encoding has a similarly simple set of vectors corresponding to trees at the OLA encoding, however takes quadratic time to compute in general.
For ``Yule-type'' trees, our OLA encoding agrees with the Phylo2Vec encoding,
in both output and computational complexity.

Although both the Phylo2Vec and the OLA encoding have the same basic idea, OLA requires only linear time.
The difference lies in the way that internal nodes are labeled.
OLA uses negative integer labels which, because they are disjoint from the positive leaf labels, can be used persistently throughout the tree construction process.
In contrast, for Phylo2Vec, internal nodes are labeled with positive integers, such that the internal node labels have to be regularly updated to avoid conflicts with the leaf labels.

Using another approach, Chauve et al.~\cite{chauve-etal} introduce the LTS vector of a tree, for ``labeling taxon sequence.'' 
This approach takes linear time to compute.
However, the construction is somewhat complex and the set of vectors corresponding to trees does not have a simple definition.

Both the Phylo2Vec and LTS papers also investigate distances between trees that are induced by the encoding. 
Specifically, Penn et al. \cite{penn-etal} examine how their Phylo2Vec-based Hamming distance relates to traditional tree metrics like the Subtree-Prune-Regraft (SPR) and Robinson-Foulds distances.
Like for the OLA encoding, a small Phylo2Vec distance can correspond to a small or a large SPR distance, although experiments with SPR random walks show an approximately linear relationship between the two distances for small numbers of SPRs.
Chauve et al.~\cite{chauve-etal} also introduce a ``HOP distance'' based on the LTS representation and analyze its relationship to NNI and SPR distances, showing that it often provides a better approximation of the SPR distance than the RF distance.

All of these encodings, including OLA, require an ordering of the set of taxa.
Although this doesn't pose a problem when calculating encodings on trees for a specific set of taxa, it does have disadvantages for the corresponding tree distances.
Namely, classical tree distances only depend on a pair of leaf-labeled trees, while the Phylo2Vec, HOP, and OLA distances depend on the ordering of the taxa.
One can obtain a well-defined distance by taking the minimum or averaging over all possible orderings, but this may be expensive to compute.

We note that while our definition of the OLA code is for rooted trees, it works equally well for unrooted trees.
This is done by picking a leaf and making it the root.
Although this requires an arbitrary choice, it introduces the same type of asymmetry as the requirement of choosing a leaf ordering, with the same associated drawbacks.

In the future, we will quantify the degree to which these encodings give natural descriptions of tree space.
We also leave these more technical questions for future work:

\begin{itemize}
  \item For a given tree shape (a rooted, bifurcating tree \emph{without} leaf labels), for a randomly chosen leaf ordering, 
  what is the expected number of negative entries in the corresponding OLA encoding? Is there any correlation with, e.g., tree balance?

  \item What is the effect of ``shuffling'' on OLA distances? Namely, for two trees with the same leaf taxa, 
  how does the OLA distance change as we shuffle the leaf order, i.e. apply a permutation to the leaf labels on both trees together?
  Can be bound the variation of OLA distance cause by shuffling?

  \item Can we extend the OLA encoding to cover multifurcating tree topologies?

\end{itemize}

In the Appendix, we 
provide a suggested extension to trees with branch lengths, 
show a more involved example of an OLA encoding, 
prove a theorem investigating random walks on trees using the OLA encoding, 
do an empirical analysis of OLA distance along random SPR and NNI walks, 
and end with an illustration of all the OLA encodings on three and four leaves.

\section{Methods}\label{sec:methods}

\subsection{OLA encoding details}
\label{sec:ola-details}

Here we give a more thorough description, with proofs, of the OLA code of a rooted, binary, leaf-ordered tree.

\subsubsection*{Theoretical description}

Given a tree with leaves labeled with $0, 1, 2, \ldots, n - 1$, we can imagine the tree being constructed from the empty tree by adding one leaf at a time in the specified order,
as shown in Figures~\ref{fig:yule-ola-code} and \ref{fig:non-yule-ola-code}.
As mentioned in the introduction, the OLA code of a rooted, bifurcating tree consists of the following main steps:

\begin{enumerate}
\item \textbf{Canonical interal label step}.
When leaf node $i$ is added, this creates a new internal node in the tree. 
Label this internal node as $-i$.

\item \textbf{OLA entry step}.
When leaf node $i$ is added, it is the sister to some previously-existing node (either a leaf node or an internal node).
Record the sister leaf's label as $a_i$, the $i$-th entry of the OLA encoding.
\end{enumerate}


\subsubsection*{Practical description}

Instead of following the steps above, where both the OLA encoding and canonical internal node labels are computed jointly at each step of the tree-contruction process,
we can make the following change.
We find it more practical to compute all canonical internal labels in the first ``phase'' of the encoding, and then locate all sister-nodes to the leaf attachments in the second ``phase.''





\subsection{Canonical internal node labels}\label{sec:canonical-labels}

Here we describe a method of labeling the internal nodes of a phylogenetic tree, whose leaves have a chosen linear ordering.
We assume this leaf ordering is encoded by assigning the labels $0, 1, 2, \ldots, n - 1$ to leaves.

\subsubsection{Label algorithm}


To obtain canonical internal node labels, we apply the following steps:

\begin{enumerate}
  \item Label internal nodes
  \begin{enumerate}
    \item Root-ward traversal: label each node (either internal or leaf) by its ``clade-founder'' label.
    At leaves, this is equal to the leaf label.
    At internal nodes, this is equal to the minimum of the two children's clade-founder labels.
    (See Figure~\ref{fig:clade-founder}.)

    \item Leaf-ward traversal: label each internal node by its ``clade-splitter'' label.
    This is equal to the maximum of the two children's clade-founder labels.
    (See Figure~\ref{fig:clade-splitter}.)

    \item Each internal node has canonical label $= -1\cdot$(clade-splitter label)
  \end{enumerate}
\end{enumerate}

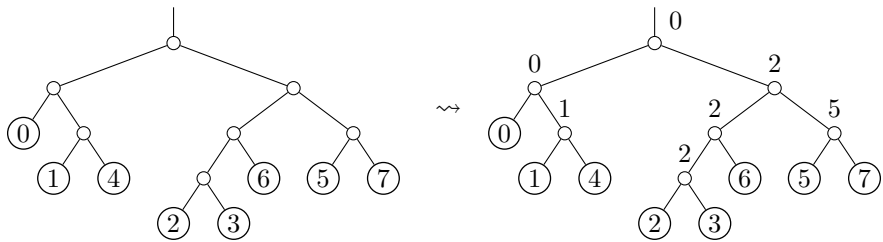
\begin{figure}[h]
  \centering
  \raisebox{-0.5\height}{\begin{tikzpicture}[scale=0.4,
    inner/.style={
      draw,circle,inner sep=0pt,minimum size=5pt
    },
    level 1/.style={sibling distance=2.0cm}
  ]
  \node {} [grow=down]
      child {node[inner] {}
        child {node[inner,left=11mm] {}
          child {node[leaf] {0}}
          child {node[inner] {}
            child {node[leaf] {1}}
            child {node[leaf] {4}}
          }
        }
        child {node[inner,right=11mm] {}
          child {node[inner,left=3mm] {}
            child {node[inner] {}
              child {node[leaf] {2}}
              child {node[leaf] {3}}
            }
            child {node[leaf] {6}}
          }
          child {node[inner,right=3mm] {}
            child {node[leaf] {5}}
            child {node[leaf] {7}}
          }
        }
      };
  \end{tikzpicture}}
  \quad$\rightsquigarrow$\quad 
  \raisebox{-0.5\height}{\begin{tikzpicture}[scale=0.4,
    level 1/.style={sibling distance=2.0cm}
  ]
  \node {} [grow=down]
      child {node[inner,label=above right:{0}] {}
        child {node[left=11mm,inner,label={0}] {}
          child {node[leaf] {0}}
          child {node[inner,label={1}] {}
            child {node[leaf] {1}}
            child {node[leaf] {4}}
          }
        }
        child {node[right=11mm,inner,label={2}] {}
          child {node[left=3mm,inner,label={2}] {}
            child {node[inner,label={2}] {}
              child {node[leaf] {2}}
              child {node[leaf] {3}}
            }
            child {node[leaf] {6}}
          }
          child {node[right=3mm,inner,label={5}] {}
            child {node[leaf] {5}}
            child {node[leaf] {7}}
          }
        }
      };
  \end{tikzpicture}}
  \caption{Assigning canonical internal node labels, step 1. Each internal node is labeled with its clade-founder label.}
  \label{fig:clade-founder}
\end{figure}

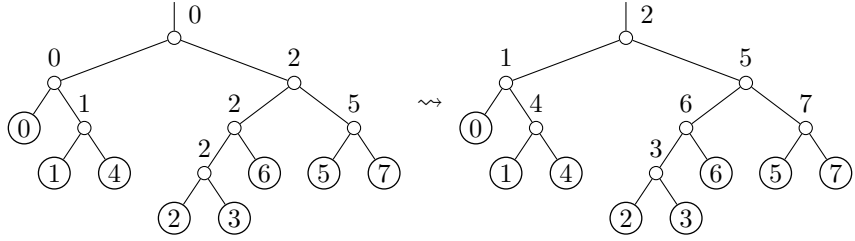
\begin{figure}[h]
  \centering
  \raisebox{-0.5\height}{\begin{tikzpicture}[scale=0.4,
    level 1/.style={sibling distance=2.0cm}
  ]
  \node {} [grow=down]
      child {node[inner,label=above right:{0}] {}
        child {node[left=11mm,inner,label={0}] {}
          child {node[leaf] {0}}
          child {node[inner,label={1}] {}
            child {node[leaf] {1}}
            child {node[leaf] {4}}
          }
        }
        child {node[right=11mm,inner,label={2}] {}
          child {node[left=3mm,inner,label={2}] {}
            child {node[inner,label={2}] {}
              child {node[leaf] {2}}
              child {node[leaf] {3}}
            }
            child {node[leaf] {6}}
          }
          child {node[right=3mm,inner,label={5}] {}
            child {node[leaf] {5}}
            child {node[leaf] {7}}
          }
        }
      };
  \end{tikzpicture}}
  $\;\rightsquigarrow\;$ 
  \raisebox{-0.5\height}{\begin{tikzpicture}[scale=0.4,
    level 1/.style={sibling distance=2.0cm}
  ]
  \node {} [grow=down]
      child {node[inner,label=above right:{2}] {}
        child {node[left=11mm,inner,label={1}] {}
          child {node[leaf] {0}}
          child {node[inner,label={4}] {}
            child {node[leaf] {1}}
            child {node[leaf] {4}}
          }
        }
        child {node[right=11mm,inner,label={5}] {}
          child {node[left=3mm,inner,label={6}] {}
            child {node[inner,label={3}] {}
              child {node[leaf] {2}}
              child {node[leaf] {3}}
            }
            child {node[leaf] {6}}
          }
          child {node[right=3mm,inner,label={7}] {}
            child {node[leaf] {5}}
            child {node[leaf] {7}}
          }
        }
      };
  \end{tikzpicture}}
  \caption{Assigning canonical internal node labels, step 2. On the right, each internal node is labeled with its clade-splitter label.
  }
  \label{fig:clade-splitter}
\end{figure}

\begin{algorithm}[H]\label{alg:canonical}
  \caption{Canonical Internal Node Labels}
  \DontPrintSemicolon
  \KwIn{binary, rooted tree $T$ with $n$ leaves, labeled $0, 1, 2, \ldots, n - 1$}
  \KwOut{$T$ with labels on internal nodes}
  \tcp{Assign clade-founder labels $\CF{v}$ to vertices}
  \For{node $v \in T$, in postorder traversal}{
    \If{$v$ is a leaf}{
      let $\CF{v} = \labelfn(v)$\;
    }{
      \Else{
        let $\CF{v} = \min(\CF{{\tt child}_1(v)}, \CF{{\tt child}_2(v)})$\;
      }
    }
  }
  \tcp{Assign clade-splitter labels $\CS{v}$ to internal vertices}
  \For{node $v \in T$, in preorder traversal}{
    \If{$v$ is not a leaf}{
      let $\CS{v} = \max(\CF{{\tt child}_1(v)}, \CF{{\tt child}_2(v)})$\;
      let $v.\labelfn = - \CS{v}$\;
    }
  }
\end{algorithm}

\begin{prop}\label{prop:canonical-label-bijection}
  Given a binary, rooted tree with leaf labels $\{0, 1, 2, \ldots, n - 1\}$, the canonical internal node labeling  (Algorithm~\ref{alg:canonical}) defines a bijective correspondence between 
  the internal nodes and the labels $\{-1, -2, \ldots, -n + 1\}$.
\end{prop}
\begin{proof}
  We prove this by induction on $n$.
  If $n = 2$, then the unique internal node is assigned label $-1$.

  Now suppose $T$ is a tree with $n$ leaves (with labels $0, \ldots, n-1$), where $n \geq 2$, and let $T^+$ denote a tree obtained by adding a leaf $v^+$ to $T$, with $\labelfn(v^+) = n$.
  By inductive hypothesis, we assume that the statement holds for $T$.
  The internal nodes of $T^+$ may be naturally identified with internal nodes of $T$, plus one additional node in $T^+$ which is the parent of leaf $n$.
  Suppose we call this parent node $p$.
  Since the statement holds for $T$ by inductive hypothesis, it suffices to show that 
  \begin{enumerate}
    \item[(a)] the naturally identified internal nodes of $T$ and $T^+$ have the same labels;
    \item[(b)] $p$ is assigned the label $-n$ in $T^+$.
  \end{enumerate}

  For part (a), first observe that the clade-founder labels must agree on all nodes shared by $T$ and $T^+$.
  Then to argue that the clade-splitter labels also agree, it suffices to consider the parent of the node $p$ (in $T^+$).
  Call this parent node $w$.
  To compute $\CS{w}$ (in $T^+$), we consider $\CF{p}$ and the clade-founder of $w$'s other child.
  The other child is not affected by the change from $T$ to $T^+$.
  For $\CF{p}$, its value is strictly smaller than $n$.
  Thus $\CF{p}$ is equal to the clade-founder label of its non-$v^+$-child, and the clade-founder label of this child agrees between $T$ and $T^+$.
  This shows the label $\CS{w}$ in $T^+$ agrees with $\CS{w}$ in $T$ as desired.
  
  For part (b),
  according to Algorithm~\ref{alg:canonical} line 8, we have $\CS{p} \geq n$ because $v$ is a child of $p$ and $\CF{v} = n$.
  However, the reverse inequality $\CS{p} \leq n$ also holds because the clade-founder labels satisfy $\CF{*} \leq \max \{\labelfn(w) : w \}$
  This shows that $\labelfn(v) = -n$ as desired.
\end{proof}
  %
  %
  %


We can give another intuitive explanation for what we are recording with the canonical internal node labels. 
Given a leaf-ordered tree, which we imagine being constructed one leaf at a time, we can naturally ``color'' each of the edges according to the step in which it was added in the construction process.
See Figure~\ref{fig:branch-decomposition}.
This coloring decomposes the tree into a collection of single-colored paths,
each of which has one end at a leaf.
If a path ends at the leaf $i$, then we label the other end (which must be an internal node) as $-i$.
This perspective makes it clear that Proposition~\ref{prop:canonical-label-bijection} holds, once we are convinced this agrees with the labels of Algorithm~\ref{alg:canonical}.
Note: Figure~\ref{fig:branch-decomposition} was inspired by a similar figure made in Chauve et al.~\cite{chauve-etal}.

\begin{figure}[h]
  \centering
  \raisebox{-0.5\height}{\begin{tikzpicture}[scale=0.5,
    inner/.style={draw,circle,inner sep=0pt,minimum size=5pt},
    level 1/.style={level distance=1cm},
  ]
  \tikzstyle{lowlight} = []
  \tikzstyle{highlight} = []
    \node (root) {} [grow=down]
      child {node[inner,fill=color3!20] (in2) {}
        child {node[inner,fill=color2!20,left=1cm] (in1) {}
          child {node[leaf,fill=color1!20] (0) {0}}
          child {node[inner,fill=color5!20] (in4) {}
            child {node[leaf,fill=color2!20] (1) {1}}
            child {node[leaf,fill=color5!20] (4) {4}}
          }
        }
        child {node[inner,fill=color6!20,right=1.2cm] (in5) {}
          child {node[inner,fill=color7!20,left=0.3cm] (in6) {}
            child {node[inner,fill=color4!20] (in3) {}
              child {node[leaf,fill=color3!20] (2) {2}}
              child {node[leaf,fill=color4!20] (3) {3}}
            }
            child {node[leaf,fill=color7!20] (6) {6}}
          }
          child {node[inner,fill=color8!20,right=0.3cm] (in7) {}
            child {node[leaf,fill=color6!20] (5) {5}}
            child {node[leaf,fill=color8!20] (7) {7}}
          }
        }
      };
    \draw[highlightedgecolor={color1!30}] (root) -- (in2) -- (in1) -- (0);
    \draw[highlightedgecolor={color2!30}] (in1) -- (in4) -- (1);
    \draw[highlightedgecolor={color3!30}] (in2) -- (in5) -- (in6) -- (in3) -- (2);
    \draw[highlightedgecolor={color4!50}] (in3) -- (3);
    \draw[highlightedgecolor={color5!30}] (in4) -- (4);
    \draw[highlightedgecolor={color6!30}] (in5) -- (in7) -- (5);
    \draw[highlightedgecolor={color7!50}] (in6) -- (6);
    \draw[highlightedgecolor={color8!50}] (in7) -- (7);
  \end{tikzpicture}}
  \qquad\qquad
  \raisebox{-0.5\height}{\begin{tikzpicture}[scale=0.5,
    inner/.style={draw,circle,inner sep=0pt,minimum size=5pt},
    level 1/.style={level distance=1cm},
  ]
  \node {} [grow=down]
      child {node[inner,label=above right:{-2}] {}
        child {node[inner,label={-1},left=1cm] {}
          child {node[leaf] {0}}
          child {node[inner,label={-4}] {}
            child {node[leaf] {1}}
            child {node[leaf] {4}}
          }
        }
        child {node[inner,label={-5},right=1.2cm] (Pu) {}
          child {node[inner,label={-6},left=0.3cm] {}
            child {node[inner,label={-3}] {}
              child {node[leaf] {2}}
              child {node[leaf] {3}}
            }
            child {node[leaf] {6}}
          }
          child {node[inner,label={-7},right=0.3cm] (P) {}
            child {node[leaf] (Pd) {5}}
            child {node[leaf] (Q) {7}}
          }
        }
      };
  \end{tikzpicture}}
  \caption{A tree with its branch decomposition highlighted (left) and its canonical internal node labels (right).}
  \label{fig:branch-decomposition}
\end{figure}
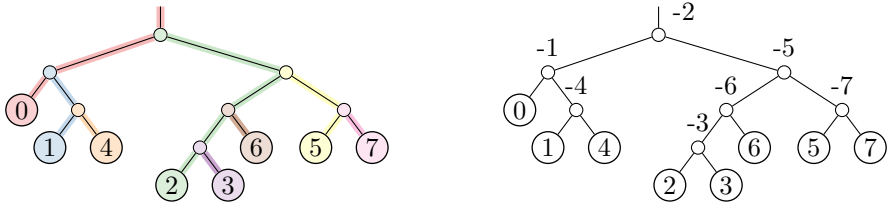


\subsection{Tree to OLA code}
\label{sec:tree-to-ola}

To produce an OLA code from a labeled binary, rooted tree, we apply the following steps:
\begin{enumerate}
  \item Apply canonical internal node labels (see Section~\ref{sec:canonical-labels}).
  \item Deconstruct the tree by removing leaves in reverse order; from leaf $n - 1$ to leaf $1$.
  \begin{enumerate}
    \item Find leaf $i$, record the label of its sister node as $a_i$.
    \item Remove leaf $i$ and its parent, connecting the leaf's sister node to its grandparent.
  \end{enumerate}
\end{enumerate}

\begin{algorithm}[H]
  \label{alg:ola-encode}
  \caption{Tree to OLA Encoding}
  \DontPrintSemicolon
  \KwIn{tree $T_n$ with $n$ leaves, labeled $0, 1, 2, \ldots, n - 1$}
  \KwOut{integer vector $(a_1, a_2, \ldots, a_{n-1}) \in \codes_{n - 1}$} 
  \BlankLine
  let $(a_1,a_2,\ldots,a_{n - 1}) = (0,0,\ldots,0)$\;
  let ${\tt leaves} = $ array with $\texttt{leaves}[i] = $ leaf with label $i$\;
  assign canonical internal node labels to $T_n$ (Algorithm~\ref{alg:canonical})\;
  \tcp{Deconstruct tree by removing leaves in reverse order,
  recording sister node of each leaf}
  \For{label $i = n - 1$ \KwTo $0$}{
    let $v = {\tt leaves}[i]$\;
    let $s = $ sister node of $v$\;
    let $a_i = s.\labelfn$\;
    \tcp{Remove leaf $v$ and its parent, after connecting sister node $s$ to its grandparent}
    let ${\tt parent}(s) = {\tt parent}({\tt parent}(v))$\;
    delete $v$, ${\tt parent}(v)$\;
  }
  \Return vector $(a_1, a_2, \ldots, a_{n - 1})$\;
\end{algorithm}

\begin{thm}
  The OLA encoding (Algorithm~\ref{alg:ola-encode}) defines a bijection 
  \[
    \Phi : \RBThat_n \to \codes_{n - 1}
  \]
  for every $n \geq 2$.
\end{thm}
\begin{proof}
  We prove this by induction on $n$.
  For $n = 2$, there is a single rooted, binary tree which is encoded as {\tt (0)}.

  Now suppose that $n > 2$, and assume by inductive hypothesis that the OLA encoding is a bijection on trees with $n - 1$ leaves.
  Since $\RBThat_n$ and $\codes_{n - 1}$ both have cardinality equal to $(2n - 3)!!$, it suffices to show that the OLA encoding function is surjective.
  Choose an arbitrary integer vector $\mathbf{v} = (a_1, a_2, \ldots, a_{n - 2}, a_{n - 1})$ in $\codes_{n - 1}$.
  Let $\mathbf{v}_{n - 1} = (a_1, \ldots, a_{n - 2})$ denote the vector obtained by deleting the last entry.
  By our inductive hypothesis, there is a tree $T_{n - 1}$ whose OLA encoding is $\mathbf{v}_{n - 1}$.
  Using $T_{n - 1}$, we describe how to construct a tree $T$ whose OLA encoding is $\mathbf{v}$.

  Suppose we assign canonical internal node labels (Algorithm~\ref{alg:canonical}) to $T_{n - 1}$.
  Then by Proposition~\ref{prop:canonical-label-bijection}, there is a unique internal node of $T_{n - 1}$ whose canonical label is $a_{n - 1}$.
  Let $T$ be the tree obtained from $T_{n - 1}$ by subdividing the parent edge of the $a_{n - 1}$-node, and adding a new leaf under the new internal node.
  In $T$, the $a_{n - 1}$-node has the same canonical label, since the clade under this node doesn't change.
  Thus the last entry of the OLA vector $\Phi(T)$ is equal to $a_{n - 1}$, as desired.
  Finally, all the previous entries in $\Phi(T)$ agree with $\mathbf{v}$ by construction, due to Proposition~\ref{prop:leaf-del-stable}.
\end{proof}

In the following statement, if $T$ is a leaf-ordered tree with $n$ nodes, the notation $T|_k$ denotes the tree obtained by pruning all nodes from $T$ that are larger than $k$, and ``collapsing'' nodes which have only one child in the result.

\begin{prop}\label{prop:leaf-del-stable}
  The OLA encoding (Algorithm~\ref{alg:ola-encode}) is stable under last-leaf-deletion.
  In other words, if $T$ has OLA vector ${\tt (a_0,a_1,\ldots,a_{n - 1})}$, then $T|_k$ has OLA vector ${\tt (a_0,a_1,\ldots,a_{k - 1})}$ for any $k < n$.
\end{prop}
\begin{proof}
  This is clear from the for-loop in lines 4-9 of Algorithm~\ref{alg:ola-encode}.
\end{proof}

\subsection{OLA code to tree}
\label{sec:ola-to-tree}

To produce a tree from an OLA code, apply the following steps.
Suppose we are given an OLA code $(a_0, a_1, \ldots, a_{n - 1})$.

\begin{enumerate}
  \item Start with a rooted tree $T_0$ having a single leaf, labeled $0$. 
  \item For each $i = 1, \ldots, n$, construct $T_{i}$ from $T_{i-1}$ as follows.
  \begin{enumerate}
    \item Suppose $i$-th entry is $a_i$. Find the node $n_i$ in tree $T_{i-1}$ labeled $a_i$.
    \item Subdivide the parent-edge of $n_i$ with a new internal node. Give this new node the label $-i$.
    \item Add a leaf node under $-i$, with label $i$. Call this resulting tree $T_i$.
  \end{enumerate}
  \item Return the final tree $T_n$ (forgetting internal labels)
\end{enumerate}

\begin{algorithm}[H]
  \label{alg:ola-decode}
  \caption{OLA to Tree Decoding Algorithm}
  \DontPrintSemicolon
  \KwIn{integer vector $(a_1, a_2, \ldots, a_{n - 1}) \in \codes_{n - 1}$}
  \KwOut{tree $T_n$ with $n$ leaves}
  \BlankLine
  let $T_0 = $ a rooted tree with a root node $\rho$ and a leaf $v_0$ labeled $0$\;
  let ${\tt nodes} = $ array of length $2n - 1$\;
  let {\tt nodes[0]} $ = v_0$\;
  \For{$i = 1$ \KwTo $n - 1$}{
    let $T_i = T_{i - 1}$\; 
    let $v = {\tt nodes}[a_i]$\;
    subdivide the parent-edge of $v$ with a new internal node $w$\;
    add new leaf node $u$ as child of $w$\;
    let ${\tt nodes}[-i] = w$, and ${\tt nodes}[i] = u$\;
  }
  \Return $T_{n - 1}$\; 
\end{algorithm}

\OLAEncodingDecoding

\begin{proof}
  We prove the claim by induction on $n$.
  For $n = 2$, both $\RBThat_2$ and $\codes_1$ contain a single element, and it is straightforward to check that we get bijections between them.

  Now suppose $n \geq 3$.
  For an integer vector $\mathbf{v} = (a_1, a_2, \ldots, a_{n - 1})$, we first claim that the last entry of $\Phi(\Psi(\mathbf{v}))$ is equal to $a_{n - 1}$.
  In the tree $\Phi(\mathbf{v})$, by construction (Algorithm~\ref{alg:ola-decode} lines 6-9) a neighborhood of leaf $n - 1$ looks like the figure below.
  \[
    \begin{tikzpicture}[scale=0.6,
      level 1/.style={level distance=1cm},
    ]
    \node {$\vdots$} [grow=down,]
      child {node[inner,label=above right:{-$n$+1}] {}
        child {node[inner,label=above left:{$a_{n - 1}$}] {}
            child {node[left=-2.0mm] {$\iddots$}}
          child {node[right=-1.0mm] {$\ddots$}}
        }
        child {node[leaf] {\small $n$-1}}
      };
    \end{tikzpicture}
  \]
  Then by construction of the OLA encoding, the last entry of $\Phi(\Psi(\mathbf{v}))$ only depends on this local picture, and results in last entry $a_{n - 1}$, as claimed.

  To confirm that the remaining entries $a_1, \ldots, a_{n - 2}$ of $\mathbf{v}$ agree with $\Phi(\Psi(\mathbf{v}))$, we apply our inductive hypothesis, along with the property that $\Phi$ is stable under last-leaf-deletion, Proposition~\ref{prop:leaf-del-stable}.
\end{proof}

\subsection{Running time}

We show that converting a tree to an OLA code, and vice versa, takes linear time in the number of leaves.

\begin{thm}\label{thm:can-label-linear}
  The process of assigning canonical internal node labels to a tree with $n$ leaves has time complexity $O(n)$.
\end{thm}
\begin{proof}
  Algorithm~\ref{alg:canonical} consists of two traversals of the tree, and each step of the traversal requires a constant-time computation.
\end{proof}

\begin{thm}\label{thm:ola-linear}
  \hfill
  \begin{enumerate}[(a)]
    \item
    The process of encoding a tree with $n$ leaves to its OLA code has time complexity $O(n)$.

    \item
    The process of decoding an OLA code of length $n - 1$ to a tree has time complexity $O(n)$.

  \end{enumerate}
\end{thm}
\begin{proof}
  \begin{enumerate}[(a)]
    \item The first step of the OLA encoding is assigning canonical internal node labels, which takes linear time (Theorem~\ref{thm:can-label-linear}).

    The next step iterates through the leaves of the tree, in reverse order from their given ordering (i.e. largest-to-smallest).
    This can be achieved efficiently by constructing a label-to-leaf array in a single tree traversal.
    Finding the sister node (line 6) takes linear time, as does the tree modification operation (lines 8-9).

    \item The OLA decoding (Algorithm~\ref{alg:ola-decode}) consists of initializing a linear-length array and a single linear-length loop, so it suffices to explain why each interation of the loop is constant-time.
    This is clear because the tree modification operations (lines 7-8) take linear time, as does updating the {\tt nodes} array (line 9).
    Copying the previous tree in the loop (line 5) can be done by reference.
  \end{enumerate}
\end{proof}

\subsection{Code details}

We have implemented the linear-time OLA encoding and decoding algorithms in Python, using the ETE3 library~\cite{ete3} and minimal other dependencies.
The package is available at \url{https://github.com/matsengrp/ola-encoding}.

\subsection*{Acknowledgements}

We thank Matthew Penn and Samir Bhatt for helpful conversations.
Our implementation of the OLA encoding and decoding algorithms use the ETE3~\cite{ete3} library.

This work was supported through US National Institutes of Health grant AI162611.
Scientific Computing Infrastructure at Fred Hutch was funded by ORIP grant S10OD028685.
Dr.\ Matsen is an Investigator of the Howard Hughes Medical Institute.

AI-assistance disclosure: ChatGPT was used to help produce a template for  typesetting pseudocode using the \texttt{algorithm2e} package.

\bibliographystyle{plain}
\bibliography{main-ref}

\begin{thebibliography}{10}

\bibitem{aldous}
David~J. Aldous.
\newblock Mixing time for a {M}arkov chain on cladograms.
\newblock {\em Combin. Probab. Comput.}, 9(3):191--204, 2000.

\bibitem{cayley}
A.~Cayley.
\newblock A theorem on trees.
\newblock {\em Quart. J.}, 23:376--378, 1888.

\bibitem{chauve-etal}
Cedric Chauve, Caroline Colijn, and Louxin Zhang.
\newblock A vector representation for phylogenetic trees.
\newblock {\em Phil. Trans. R. Soc. B}, 380, 2025.

\bibitem{diaconis-holmes-walks}
Persi Diaconis and Susan Holmes.
\newblock Random walks on trees and matchings.
\newblock {\em Electron. J. Probab.}, 7:Paper 6, 17 pp., 2002.

\bibitem{diaconis-holmes-matchings}
Persi~W. Diaconis and Susan~P. Holmes.
\newblock Matchings and phylogenetic trees.
\newblock {\em Proc. Natl. Acad. Sci. USA}, 95(25):14600--14602, 1998.

\bibitem{erdos-szekely}
P{\'e}ter~L. Erd{\"o}s and L{\'a}szl{\'o}~A. Sz{\'e}kely.
\newblock Applications of antilexicographic order. {I}: {An} enumerative theory
  of trees.
\newblock {\em Adv. Appl. Math.}, 10(4):488--496, 1989.

\bibitem{ete3}
Jaime Huerta-Cepas, François Serra, and Peer Bork.
\newblock {ETE 3: Reconstruction, Analysis, and Visualization of Phylogenomic
  Data}.
\newblock {\em Molecular Biology and Evolution}, 33(6):1635--1638, 2016.

\bibitem{motwani-raghavan}
Rajeev Motwani and Prabhakar Raghavan.
\newblock {\em Randomized algorithms}.
\newblock Cambridge: Cambridge Univ. Press, 1995.

\bibitem{penn-etal}
Matthew~J Penn, Neil Scheidwasser, Mark~P Khurana, David~A Duchêne, Christl~A
  Donnelly, and Samir Bhatt.
\newblock {Phylo2Vec}: a vector representation for binary trees.
\newblock {\em Systematic Biology}, 2024.

\bibitem{prufer}
H.~Pr{\"u}fer.
\newblock Neuer {Beweis} eines {Satzes} {\"u}ber {Permutationen}.
\newblock {\em Arch. der Math. u. Phys. (3)}, 27:142--144, 1918.

\bibitem{semple-steel}
Charles Semple and Mike Steel.
\newblock {\em Phylogenetics}, volume~24 of {\em Oxf. Lect. Ser. Math. Appl.}
\newblock Oxford: Oxford University Press, 2003.

\bibitem{xie2025phylovae}
Tianyu Xie, Harry Richman, Jiansi Gao, Frederick~A Matsen, and Cheng Zhang.
\newblock Phylo{VAE}: Unsupervised learning of phylogenetic trees via
  variational autoencoders.
\newblock In {\em The Thirteenth International Conference on Learning
  Representations}, 2025.

\end{thebibliography}


\clearpage

\appendix

\begin{center}
  \textbf{\huge Appendix}
\end{center}

\section{OLA code for trees with branch lengths}

The OLA encoding can be easily extended to encode phylogenetic trees with branch lengths.
This can be achieved by adding two additional rows underneath the original OLA code, with entries consisting of positive real numbers.
These real numbers prescribe the lengths of the two new branches that form when a new leaf is attached.

Rather than spell out all the details here, we illustrate a possible branch-length OLA encoding with an example.
Figure~\ref{fig:branch-lengths} shows a tree with branch lengths, as it is constructed one leaf at a time.
This tree is encoded by the array
\[
  \begin{pmatrix}
    \texttt{0} & \texttt{-1} & \mathtt{1} & \mathtt{2} \\
    \tikzmark{1up} 0.1 & \tikzmark{2up} 0.6 & \tikzmark{3up} 0.2 & \tikzmark{4up} 0.4 \\
    0.7 \tikzmark{1down} & 0.3 \tikzmark{2down} & 0.1 \tikzmark{3down} & 0.9 \tikzmark{4down}
  \end{pmatrix},
  \DrawBox[color1!80]{1up}{1down}
  \DrawBox[color2!80]{2up}{2down}
  \DrawBox[color3!80]{3up}{3down}
  \DrawBox[color4!80]{4up}{4down}
\]
where the colors highlight which lengths correspond to which branches.
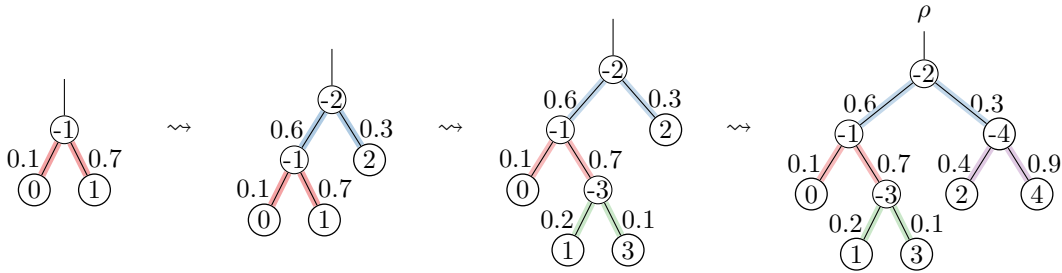
\begin{figure}[h]
  \centering
  \makebox[\textwidth][c]{

  \raisebox{-0.5\height}{\begin{tikzpicture}[scale=0.4,
    level distance=2cm,
    level 1/.style={sibling distance=2cm},
  ]
    \node {} [grow=down]
      child {node[inner,label=above right:{}] (in1) {-1}
        child {node[leaf] (0) {0} edge from parent node[left,draw=none] {0.1}}
        child {node[leaf] (1) {1} edge from parent node[right,draw=none] {0.7}}
      };
    \draw[highlightedgecolor={color1!40}] (0) -- (in1) -- (1);
  \end{tikzpicture}}
  $\quad\rightsquigarrow\quad$
  \raisebox{-0.5\height}{\begin{tikzpicture}[scale=0.4,
    level distance=2cm,
    level 2/.style={sibling distance=2.5cm},
    level 3/.style={sibling distance=2cm},
    level 4/.style={sibling distance=2cm},
  ]
    \node {} [grow=down]
      child {node[inner,minimum size=7pt,label=above right:{}] (in2) {-2}
        child {node[inner,minimum size=7pt,label={}] (in1) {-1}
          child {node[leaf] (0) {0} edge from parent node[left,draw=none] {0.1}}
          child {node[leaf] (1) {1} edge from parent node[right,draw=none] {0.7}}
          edge from parent node[left,draw=none] {0.6}
        }
        child {node[leaf] (2) {2} edge from parent node[right,draw=none] {0.3}}
      };
    \draw[highlightedgecolor={color1!40}] (0) -- (in1) -- (1);
    \draw[highlightedgecolor={color2!40}] (in1) -- (in2) -- (2);
  \end{tikzpicture}}
  \quad$\rightsquigarrow$\quad
  \raisebox{-0.5\height}{\begin{tikzpicture}[scale=0.4,
    level distance=2cm,
    level 2/.style={sibling distance=3.5cm},
    level 3/.style={sibling distance=2.5cm},
    level 4/.style={sibling distance=2cm},
  ]
  \node {} [grow=down]
      child {node[inner,label=above right:{}] (in2) {-2}
        child {node[inner,label={}] (in1) {-1}
          child {node[leaf] (0) {0} edge from parent node[left,draw=none] {0.1}}
          child {node[inner,label={}] (in3) {-3}
            child {node[leaf] (1) {1} edge from parent node[left,draw=none] {0.2}}
            child {node[leaf] (3) {3} edge from parent node[right,draw=none] {0.1}}
            edge from parent node[right,draw=none] {0.7}
          }
          edge from parent node[left,draw=none] {0.6}
        }
        child {node[leaf] (2) {2} edge from parent node[right,draw=none] {0.3}}
      };
    \draw[highlightedgecolor={color1!30}] (0) -- (in1) -- (in3);
    \draw[highlightedgecolor={color2!30}] (in1) -- (in2) -- (2);
    \draw[highlightedgecolor={color3!30}] (1) -- (in3) -- (3);
  \end{tikzpicture}}
  \quad$\rightsquigarrow$\quad
  \raisebox{-0.5\height}{\begin{tikzpicture}[scale=0.4,
    level distance=2cm,
    level 2/.style={sibling distance=5.0cm},
    level 3/.style={sibling distance=2.5cm},
    level 4/.style={sibling distance=2cm},
  ]
  \node {$\rho$} [grow=down]
      child {node[inner,label=above right:{}] (in2) {-2}
        child {node[inner,label={}] (in1) {-1}
          child {node[leaf] (0) {0} edge from parent node[left,draw=none] {0.1}}
          child {node[inner,label={}] (in3) {-3}
            child {node[leaf] (1) {1} edge from parent node[left,draw=none] {0.2}}
            child {node[leaf] (3) {3} edge from parent node[right,draw=none] {0.1}}
            edge from parent node[right,draw=none] {0.7}
          }
          edge from parent node[left,draw=none] {0.6}
        }
        child {node[leaf] (in4) {-4} 
          child {node[leaf] (2) {2} edge from parent node[left,draw=none] {0.4}}
          child {node[leaf] (4) {4} edge from parent node[right,draw=none] {0.9}}
          edge from parent node[right,draw=none] {0.3}}
      };
    \draw[highlightedgecolor={color1!30}] (0) -- (in1) -- (in3);
    \draw[highlightedgecolor={color2!30}] (in1) -- (in2) -- (in4);
    \draw[highlightedgecolor={color3!30}] (1) -- (in3) -- (3);
    \draw[highlightedgecolor={color4!30}] (2) -- (in4) -- (4);
  \end{tikzpicture}}
  }
  \caption{OLA encoding of a binary tree with branch lengths.}
  \label{fig:branch-lengths}
\end{figure}
Alternatively, the tree with branch lengths could be encoded in ``one-line format'' as \texttt{(0:\colorbox{color1!20}{0.1:0.7}, -1:\colorbox{color2!20}{0.6:0.3}, 1:\colorbox{color3!20}{0.2:0.1}, 2:\colorbox{color4!20}{0.4:0.9})}

\section{Example OLA encoding}\label{sec:ola-ex-large}

We demonstrate the construction with a more involved example of a phylogenetic tree, compared to the example in the introduction,
shown in Figure~\ref{fig:ola-ex-larger}.

\begin{figure}[h]
  \centering
  \raisebox{-0.5\height}{\begin{tikzpicture}[scale=0.6,
  ]
    \node[tnode] {$\rho$} [grow=down]
      child {node[leaf] {0}};
  \end{tikzpicture}}
  $\quad\rightsquigarrow\quad$
  \raisebox{-0.5\height}{\begin{tikzpicture}[scale=0.5,
    level 1/.style={level distance=1cm},
  ]
    \node {$\rho$} [grow=down]
      child {node[lowlight,inner,label=above right:{-1}] (P) {}
        child {node[leaf,highlight] {0}}
        child {node[leaf,lowlight] (Q) {1}}
      };
    \draw[lowlightedge] (P) -- (Q);
    \draw (P) -- (Q);
  \end{tikzpicture}}
  $\quad\rightsquigarrow\quad$
  \raisebox{-0.5\height}{\begin{tikzpicture}[scale=0.5,
    level 1/.style={level distance=1cm},
  ]
    \node {$\rho$} [grow=down]
      child {node[lowlight,inner,label=above right:{-2}] (P) {}
        child {node[highlight,inner,label={-1}] {}
          child {node[leaf] {0}}
          child {node[leaf] {1}}
        }
        child {node[leaf,lowlight] (Q) {2}}
      };
    \draw[lowlightedge] (P) -- (Q);
    \draw (P) -- (Q);
  \end{tikzpicture}}
  $\quad\rightsquigarrow\quad$
  \raisebox{-0.5\height}{\begin{tikzpicture}[scale=0.5,
    level 1/.style={level distance=1cm},
  ]
    \node {$\rho$} [grow=down]
      child {node[inner,label=above right:{-2}] {}
        child {node[inner,label={-1},left=0.4cm] {}
          child {node[leaf] {0}}
          child {node[leaf] {1}}
        }
        child {node[tnode,lowlight,inner,label={-3},right=0.4cm] (P) {}
          child {node[leaf,highlight] {2}}
          child {node[leaf,lowlight] (Q) {3}}
        }
      };
    \draw[lowlightedge] (P) -- (Q);
    \draw (P) -- (Q);
  \end{tikzpicture}}
  $\quad\rightsquigarrow\quad$\\[0.6cm]
  \raisebox{-0.5\height}{\begin{tikzpicture}[scale=0.5,
    level 1/.style={level distance=1cm},
  ]
  \node {$\rho$} [grow=down]
      child {node[inner,label=above right:{-2}] {}
        child {node[inner,label={-1},left=0.8cm] {}
          child {node[leaf] {0}}
          child {node[inner,lowlight,label={-4}] (P) {}
            child {node[leaf,highlight] {1}}
            child {node[leaf,lowlight] (Q) {4}}
          }
        }
        child {node[inner,label={-3},right=0.5cm] {}
          child {node[leaf] {2}}
          child {node[leaf] {3}}
        }
      };
    \draw[lowlightedge] (P) -- (Q);
    \draw (P) -- (Q);
  \end{tikzpicture}}
  $\quad\rightsquigarrow\quad$
  \raisebox{-0.5\height}{\begin{tikzpicture}[scale=0.5,
    level 1/.style={level distance=1cm},
  ]
  \node {$\rho$} [grow=down]
      child {node[inner,label=above right:{-2}] (Pu) {}
        child {node[inner,label={-1},left=0.8cm] {}
          child {node[leaf] {0}}
          child {node[inner,label={-4}] {}
            child {node[leaf] {1}}
            child {node[leaf] {4}}
          }
        }
        child {node[lowlight,inner,label={-5},right=0.8cm] (P) {}
          child {node[highlight,inner,label={-3}] (Pd) {}
            child {node[leaf] {2}}
            child {node[leaf] {3}}
          }
          child {node[leaf,lowlight] (Q) {5}}
        }
      };
  \end{tikzpicture}}
  $\quad\rightsquigarrow\quad$\\[0.6cm]
  \raisebox{-0.5\height}{\begin{tikzpicture}[scale=0.5,
    level 1/.style={level distance=1cm},
  ]
  \node {$\rho$} [grow=down]
      child {node[inner,label=above right:{-2}] {}
        child {node[inner,label={-1},left=1cm] {}
          child {node[leaf] {0}}
          child {node[inner,label={-4}] {}
            child {node[leaf] {1}}
            child {node[leaf] {4}}
          }
        }
        child {node[inner,label={-5},right=1.0cm] (Pu) {}
          child {node[lowlight,inner,label={-6}] (P) {}
            child {node[highlight,inner,label={-3}] (Pd) {}
              child {node[leaf] {2}}
              child {node[leaf] {3}}
            }
            child {node[leaf,lowlight] (Q) {6}}
          }
          child {node[leaf] {5}}
        }
      };
  \end{tikzpicture}}
  $\quad\rightsquigarrow\quad$
  \raisebox{-0.5\height}{\begin{tikzpicture}[scale=0.5,
    level 1/.style={level distance=1cm},
  ]
  \node {$\rho$} [grow=down]
      child {node[inner,label=above right:{-2}] {}
        child {node[inner,label={-1},left=1cm] {}
          child {node[leaf] {0}}
          child {node[inner,label={-4}] {}
            child {node[leaf] {1}}
            child {node[leaf] {4}}
          }
        }
        child {node[inner,label={-5},right=1.2cm] (Pu) {}
          child {node[inner,label={-6},left=0.3cm] {}
            child {node[inner,label={-3}] {}
              child {node[leaf] {2}}
              child {node[leaf] {3}}
            }
            child {node[leaf] {6}}
          }
          child {node[lowlight,inner,label={-7},right=0.3cm] (P) {}
            child {node[leaf,highlight] (Pd) {5}}
            child {node[leaf,lowlight] (Q) {7}}
          }
        }
      };
  \end{tikzpicture}}
  \caption{Constructing the tree with OLA code $\texttt{(0,-1,2,1,-3,-3,5)}$.}
  \label{fig:ola-ex-larger}
\end{figure}

\section{Tree distance and random walks}

\subsection{Generating random trees}

A common problem in phylogenetics is to efficiently traverse the space of all (binary) phylogenetic trees on a given number of leaves.

The OLA code allows us to define a random walk on the set of rooted, binary trees $\RBT_n$ by applying the OLA decoding function to a Gibbs random walk on the set of vectors $\codes_{n - 1}$.
By ``Gibbs random walk,'' we mean one step in the random walk consists of choosing a random index $i$, and changing the entry at $i$ to any valid integer uniformly, so that the outcome lies in $\codes_{n - 1}$.
It is natural to ask: how quickly do we expect this random walk to become well-dispersed over tree space?


\begin{thm}
  Let $t^*$ denote the number of steps in a random OLA walk after which the resulting probability distribution on $\RBT_n$ is uniform.
  Then $\EE(t^*) = \Theta(n \log n)$.
\end{thm}

\begin{proof}
  This is a direct application of the expected time of the coupon collection problem~\cite[Chapter 3.6]{motwani-raghavan}.

  The OLA code consists of $n - 1$ entries, and each step in an OLA walk (Gibbs walk) consists of choosing one entry and changing the value of that entry to a uniformly-random chosen value (out of the set of valid values, $-i < a_i < i$).
  The state of the walk is independent of the starting point exactly when every entry between $a_2$ and $a_{n - 1}$ has been chosen to be ``mixed.''
  (We ignore $a_1$, since this entry is always $0$.)
  This situation is equivalent to the coupon collector problem: suppose there are $n - 2$ coupon types, and one wishes to collect all types.
  Each coupon has a uniformly-random type.
  Then a standard result of probability theory states that the expected number of coupons needed is $\Theta((n - 2) \log(n - 2))$, which is asymptotically equivalent to $n \log n$.
\end{proof}

Diaconis and Holmes~\cite{diaconis-holmes-walks} consider a different method for randomly perturbing binary trees that also mixes in $n \log n$ time, up to a constant factor.
Earlier, Aldous~\cite{aldous} considers randomly perturbing binary trees by pruning and regrafting a single leaf at each step.
He shows that this process mixes more slowly, in polynomial time between $O(n^2)$ and $O(n^3)$.

\subsection{OLA distance along random SPR walks and NNI walks}

In Section~\ref{sec:ola-distance-compare}, we investigated bounds on the OLA distance of two trees which differ by a single NNI move or a single SPR move.
In general, we may also want to consider phylogenetic trees that are farther apart in NNI- or SPR-distance.

In Figures~\ref{fig:spr-walk-dist-100} and \ref{fig:nni-walk-dist-100}, we consider running an SPR-random-walk (respectively, NNI-random-walk) in which a tree is repeatedly perturbed by a uniformly selected random SPR or NNI move.
We can observe in the figures that the OLA distance increases approximately linearly along the walk until it levels off after sufficiently many moves.
This leveling-off likely corresponds to a leveling-off of the SPR distance as well, in which a collection of random SPR moves can be compressed into a smaller number of SPR moves.
For example, once the random walk reaches the SPR diameter of tree space, the SPR distance cannot increase any further.

Compare to~\cite[Figure 5]{penn-etal} for the Phylo2Vec distance.

\begin{figure}[h]
  \centering
  \includegraphics*[scale=0.8]{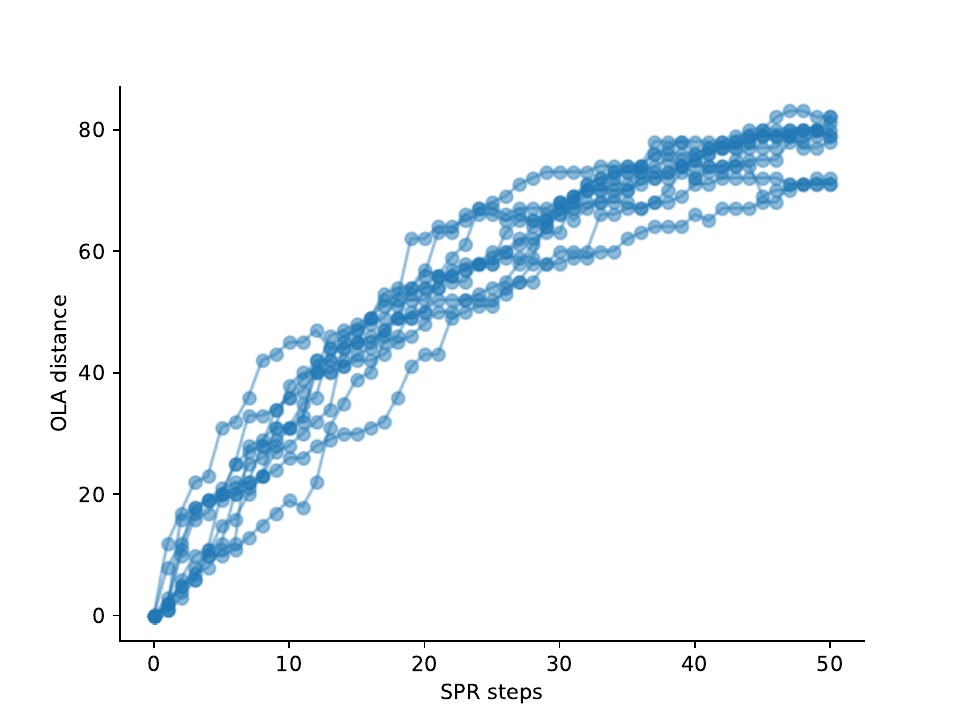}
  \caption{Comparing OLA distance for random SPR walks on trees with $100$ leaves.
  Here each walk consists of choosing a random starting tree and applying a sequence of $50$ random SPR moves.
  The plot shows the OLA distance between the starting tree and each tree along the walk, for $10$ SPR walks.}
  \label{fig:spr-walk-dist-100}
\end{figure}

\begin{figure}[h]
  \centering
  \includegraphics*[scale=0.8]{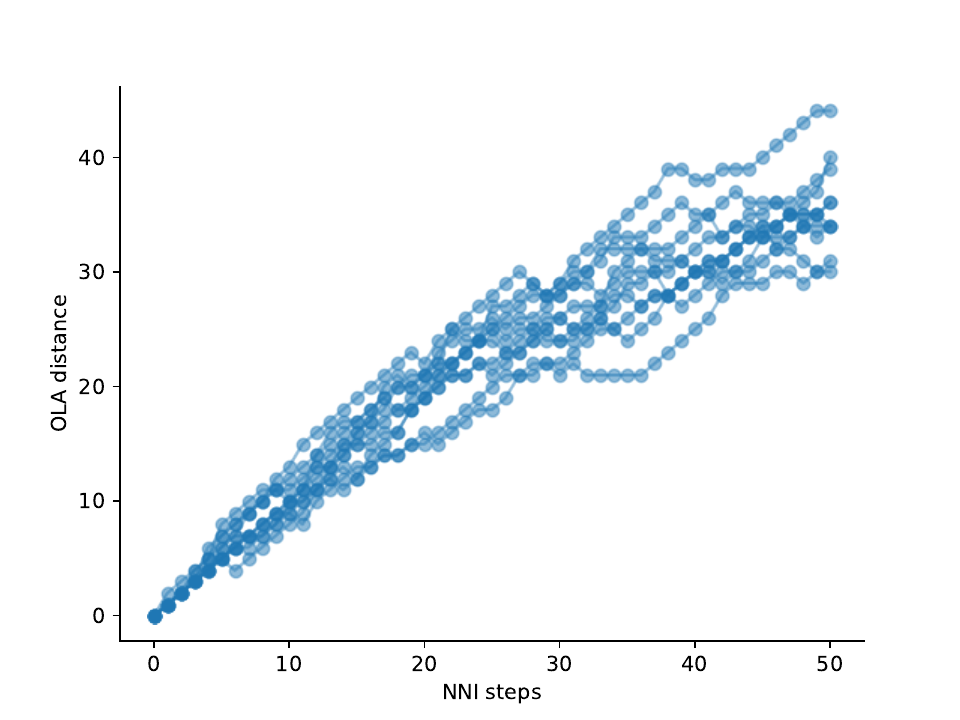}
  \caption{Comparing OLA distance for random NNI walks on trees with $100$ leaves.
  Here each walk consists of choosing a random starting tree and applying a sequence of $50$ random NNI moves.
  The plot shows the OLA distance between the starting tree and each tree along the walk, for $10$ NNI walks.}
  \label{fig:nni-walk-dist-100}
\end{figure}

\section{OLA code for trees on three or four leaves}\label{sec:4-leaves-ex}

We demonstrate the OLA encoding by showing the OLA code for all rooted, bifurcating trees on tree and four leaves,
in Figures~\ref{fig:all-ola-3} and \ref{fig:all-ola-4}.

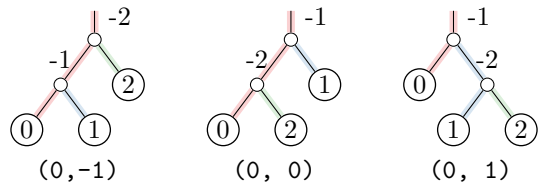
\begin{figure}[h]
  \centering
  \begin{tabular}{ccc}
    \begin{tikzpicture}[scale=0.6,
      level 1/.style={level distance=1cm},
    ]
      \node[tnode] (root) {} [grow=down,]
      child {node[inner,label=above right:{-$2$}] (in2) {}
        child {node[inner,label={-$1$}] (in1) {}
          child {node[leaf] (0) {$0$}}
          child {node[leaf] (1) {$1$}}
        }
        child {node[leaf] (2) {$2$}}
      };
      \draw[highlightedgecolor={color1!20}] (0) -- (in1) -- (in2) -- (root);
      \draw[highlightedgecolor={color2!20}] (1) -- (in1);
      \draw[highlightedgecolor={color3!20}] (2) -- (in2);
      \node[fit=(current bounding box),
      inner ysep=0mm,inner xsep=2mm] {};
    \end{tikzpicture}
    &
    \begin{tikzpicture}[scale=0.6,
      level 1/.style={level distance=1cm},
    ]
      \node[tnode] {} [grow=down]
      child {node[inner,label=above right:{-1}] (in1) {}
        child {node[inner,label={-2}] (in2) {}
          child {node[leaf] (0) {0}}
          child {node[leaf] (2) {2}}
        }
        child{node[leaf] (1) {1}}
      };
      \draw[highlightedgecolor={color1!20}] (0) -- (in2) -- (in1) -- (root);
      \draw[highlightedgecolor={color2!20}] (1) -- (in1);
      \draw[highlightedgecolor={color3!20}] (2) -- (in2);
      \node[fit=(current bounding box),
      inner ysep=0mm,inner xsep=2mm] {};
    \end{tikzpicture}
    &
    \begin{tikzpicture}[scale=0.6,
      level 1/.style={level distance=1cm},
    ]
      \node[tnode] {} [grow=down]
        child {node[inner,label=above right:{-1}] (in1) {}
          child {node[leaf] (0) {0}}
          child {node[inner,label={-$2$}] (in2) {}
            child {node[leaf] (1) {$1$}}
            child {node[leaf] (2) {$2$}}
          }
      };
      \draw[highlightedgecolor={color1!20}] (0) -- (in1) -- (root);
      \draw[highlightedgecolor={color2!20}] (1) -- (in2) -- (in1);
      \draw[highlightedgecolor={color3!20}] (2) -- (in2);
      \node[fit=(current bounding box),
      inner ysep=0mm,inner xsep=2mm] {};
    \end{tikzpicture}
    \\
    {\tt (0,-1)}
    &
    {\tt (0, 0)}
    &
    {\tt (0, 1)}
  \end{tabular}
  \caption{Phylogenetic trees with three leaves.}
  \label{fig:all-ola-3}
\end{figure}

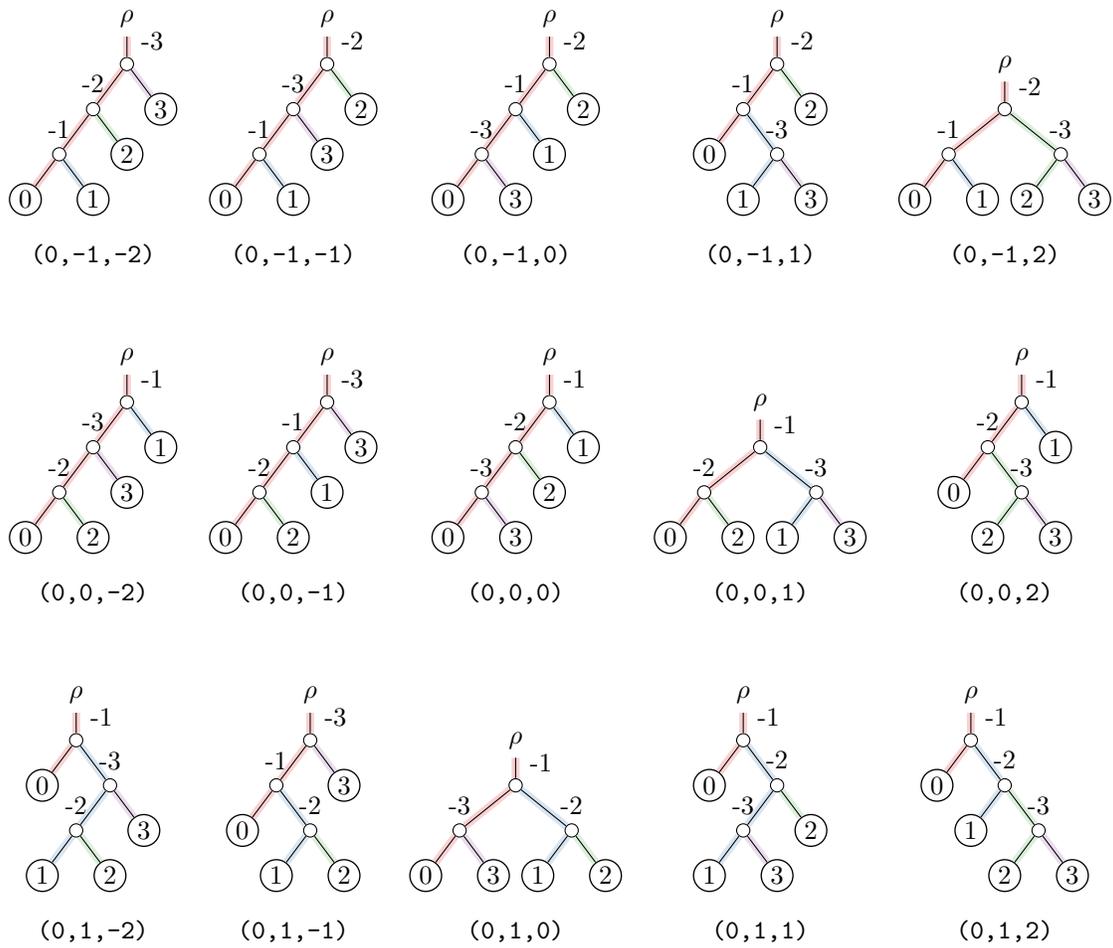
\begin{figure}[h]
  \centering
  \makebox[\textwidth][c]{
  \begin{tabular}{ccccc}
    \hline 
    \begin{tikzpicture}[scale=0.6,
      level 1/.style={level distance=1cm},
    ]
    \node (root) {$\rho$} [grow=down,]
      child {node[inner,label=above right:{-3}] (in3) {}
        child {node[inner,label={-2}] (in2) {}
          child {node[inner,label={-1}] (in1) {}
            child {node[leaf] (0) {0}}
            child {node[leaf] (1) {1}}
          }
          child {node[leaf] (2) {2}}
        }
        child {node[leaf] (3) {3}}
      };
      \draw[highlightedgecolor={color1!20}] (0) -- (in1) -- (in2) -- (in3) -- (root);
      \draw[highlightedgecolor={color2!20}] (1) -- (in1);
      \draw[highlightedgecolor={color3!20}] (2) -- (in2);
      \draw[highlightedgecolor={color4!20}] (3) -- (in3);
      \node[fit=(current bounding box),inner ysep=0.2cm,inner xsep=0] {};
    \end{tikzpicture}
    &
    \begin{tikzpicture}[scale=0.6,
      level 1/.style={level distance=1cm},
    ]
    \node (root) {$\rho$} [grow=down,]
      child {node[inner,label=above right:{-2}] (in2) {}
        child {node[inner,label={-3}] (in3) {}
          child {node[inner,label={-1}] (in1) {}
            child {node[leaf] (0) {0}}
            child {node[leaf] (1) {1}}
          }
          child {node[leaf] (3) {3}}
        }
        child {node[leaf] (2) {2}}
      };
      \draw[highlightedgecolor={color1!20}] (0) -- (in1) -- (in3) -- (in2) -- (root);
      \draw[highlightedgecolor={color2!20}] (1) -- (in1);
      \draw[highlightedgecolor={color3!20}] (2) -- (in2);
      \draw[highlightedgecolor={color4!20}] (3) -- (in3);
      \node[fit=(current bounding box),inner ysep=0.2cm,inner xsep=0] {};
    \end{tikzpicture}
    &
    \begin{tikzpicture}[scale=0.6,
      level 1/.style={level distance=1cm},
    ]
    \node (root) {$\rho$} [grow=down,]
      child {node[inner,label=above right:{-2}] (in2) {}
        child {node[inner,label={-1}] (in1) {}
          child {node[inner,label={-3}] (in3) {}
            child {node[leaf] (0) {0}}
            child {node[leaf] (3) {3}}
          }
          child {node[leaf] (1) {1}}
        }
        child {node[leaf] (2) {2}}
      };
      \draw[highlightedgecolor={color1!20}] (0) -- (in3) -- (in1) -- (in2) -- (root);
      \draw[highlightedgecolor={color2!20}] (1) -- (in1);
      \draw[highlightedgecolor={color3!20}] (2) -- (in2);
      \draw[highlightedgecolor={color4!20}] (3) -- (in3);
      \node[fit=(current bounding box),inner ysep=0.2cm,inner xsep=0] {};
    \end{tikzpicture}
    &
    \begin{tikzpicture}[scale=0.6,
      level 1/.style={level distance=1cm},
    ]
    \node (root) {$\rho$} [grow=down,]
      child {node[inner,label=above right:{-2}] (in2) {}
        child {node[inner,label={-1}] (in1) {}
          child {node[leaf] (0) {0}}
          child {node[inner,label={-3}] (in3) {}
            child {node[leaf] (1) {1}}
            child {node[leaf] (3) {3}}
          }
        }
        child {node[leaf] (2) {2}}
      };
      \draw[highlightedgecolor={color1!20}] (0) -- (in1) -- (in2) -- (root);
      \draw[highlightedgecolor={color2!20}] (1) -- (in3) -- (in1);
      \draw[highlightedgecolor={color3!20}] (2) -- (in2);
      \draw[highlightedgecolor={color4!20}] (3) -- (in3);
      \node[fit=(current bounding box),inner ysep=0.2cm,inner xsep=0] {};
    \end{tikzpicture}
    &
    \begin{tikzpicture}[scale=0.6,
      level 1/.style={level distance=1cm},
    ]
    \node (root) {$\rho$} [grow=down,]
      child {node[inner,label=above right:{-2}] (in2) {}
        child {node[left=2mm,inner,label={-1}] (in1) {}
          child {node[leaf] (0) {0}}
          child {node[leaf] (1) {1}}
        }
        child {node[right=2mm,inner,label={-3}] (in3) {}
          child {node[leaf] (2) {2}}
          child {node[leaf] (3) {3}}
        }
      };
      \draw[highlightedgecolor={color1!20}] (0) -- (in1) -- (in2) -- (root);
      \draw[highlightedgecolor={color2!20}] (1) -- (in1);
      \draw[highlightedgecolor={color3!20}] (2) -- (in3) -- (in2);
      \draw[highlightedgecolor={color4!20}] (3) -- (in3);
      \node[fit=(current bounding box),inner ysep=0.2cm,inner xsep=0] {};
    \end{tikzpicture}
    \\
    {\tt (0,-1,-2)}
    &
    {\tt (0,-1,-1)}
    &
    {\tt (0,-1,0)}
    &
    {\tt (0,-1,1)}
    &
    {\tt (0,-1,2)}
    \\[2em]
    \begin{tikzpicture}[scale=0.6,
      level 1/.style={level distance=1cm},
    ]
    \node (root) {$\rho$} [grow=down]
      child {node[inner,label=above right:{-1}] (in1) {}
        child {node[inner,label={-3}] (in3) {}
          child {node[inner,label={-2}] (in2) {}
            child {node[leaf] (0) {0}}
            child {node[leaf] (2) {2}}
          }
          child {node[leaf] (3) {3}}
        }
        child{node[leaf] (1) {1}}
      };
      \draw[highlightedgecolor={color1!20}] (0) -- (in2) -- (in3) -- (in1) -- (root);
      \draw[highlightedgecolor={color2!20}] (1) -- (in1);
      \draw[highlightedgecolor={color3!20}] (2) -- (in2);
      \draw[highlightedgecolor={color4!20}] (3) -- (in3);
      \node[fit=(current bounding box),inner ysep=0.2cm,inner xsep=0] {};
    \end{tikzpicture}
    &
    \begin{tikzpicture}[scale=0.6,
      level 1/.style={level distance=1cm},
    ]
    \node (root) {$\rho$} [grow=down]
      child {node[inner,label=above right:{-3}] (in3) {}
        child {node[inner,label={-1}] (in1) {}
          child {node[inner,label={-2}] (in2) {}
            child {node[leaf] (0) {0}}
            child {node[leaf] (2) {2}}
          }
          child{node[leaf] (1) {1}}
        }
        child {node[leaf] (3) {3}}
      };
      \draw[highlightedgecolor={color1!20}] (0) -- (in2) -- (in1) -- (in3) -- (root);
      \draw[highlightedgecolor={color2!20}] (1) -- (in1);
      \draw[highlightedgecolor={color3!20}] (2) -- (in2);
      \draw[highlightedgecolor={color4!20}] (3) -- (in3);
      \node[fit=(current bounding box),inner ysep=0.2cm,inner xsep=0] {};
    \end{tikzpicture}
    &
    \begin{tikzpicture}[scale=0.6,
      level 1/.style={level distance=1cm},
    ]
    \node (root) {$\rho$} [grow=down]
      child {node[inner,label=above right:{-1}] (in1) {}
        child {node[inner,label={-2}] (in2) {}
          child {node[inner,label={-3}] (in3) {}
            child {node[leaf] (0) {0}}
            child {node[leaf] (3) {3}}
          }
          child {node[leaf] (2) {2}}
        }
        child{node[leaf] (1) {1}}
      };
      \draw[highlightedgecolor={color1!20}] (0) -- (in3) -- (in2) -- (in1) -- (root);
      \draw[highlightedgecolor={color2!20}] (1) -- (in1);
      \draw[highlightedgecolor={color3!20}] (2) -- (in2);
      \draw[highlightedgecolor={color4!20}] (3) -- (in3);
      \node[fit=(current bounding box),inner ysep=0.2cm,inner xsep=0] {};
    \end{tikzpicture}
    &
    \begin{tikzpicture}[scale=0.6,
      level 1/.style={level distance=1cm},
    ]
    \node (root) {$\rho$} [grow=down]
      child {node[inner,label=above right:{-1}] (in1) {}
        child {node[left=0.2cm,inner,label={-2}] (in2) {}
          child {node[leaf] (0) {0}}
          child {node[leaf] (2) {2}}
        }
        child {node[right=0.2cm,inner,label={-3}] (in3) {}
          child {node[leaf] (1) {1}}
          child {node[leaf] (3) {3}}
        }
      };
      \draw[highlightedgecolor={color1!20}] (0) -- (in2) -- (in1) -- (root);
      \draw[highlightedgecolor={color2!20}] (1) -- (in3) -- (in1);
      \draw[highlightedgecolor={color3!20}] (2) -- (in2);
      \draw[highlightedgecolor={color4!20}] (3) -- (in3);
      \node[fit=(current bounding box),inner ysep=0.2cm,inner xsep=0] {};
    \end{tikzpicture}
    &
    \begin{tikzpicture}[scale=0.6,
      level 1/.style={level distance=1cm},
    ]
    \node (root) {$\rho$} [grow=down]
      child {node[inner,label=above right:{-1}] (in1) {}
        child {node[inner,label={-2}] (in2) {}
          child {node[leaf] (0) {0}}
          child {node[inner,label={-3}] (in3) {}
            child {node[leaf] (2) {2}}
            child {node[leaf] (3) {3}}
          }
        }
        child{node[leaf] (1) {1}}
      };
      \draw[highlightedgecolor={color1!20}] (0) -- (in2) -- (in1) -- (root);
      \draw[highlightedgecolor={color2!20}] (1) -- (in1);
      \draw[highlightedgecolor={color3!20}] (2) -- (in3) -- (in2);
      \draw[highlightedgecolor={color4!20}] (3) -- (in3);
      \node[fit=(current bounding box),inner ysep=0.2cm,inner xsep=0] {};
    \end{tikzpicture}
    \\
    {\tt (0,0,-2)}
    &
    {\tt (0,0,-1)}
    &
    {\tt (0,0,0)}
    &
    {\tt (0,0,1)}
    &
    {\tt (0,0,2)}
    \\[2em]
    \begin{tikzpicture}[scale=0.6,
      level 1/.style={level distance=1cm},
    ]
      \node (root) {$\rho$} [grow=down]
      child {node[inner,label=above right:{-1}] (in1) {}
        child {node[leaf] (0) {0}}
        child {node[inner,label={-3}] (in3) {}
          child {node[inner,label={-2}] (in2) {}
            child {node[leaf] (1) {1}}
            child {node[leaf] (2) {2}}
          }
          child {node[leaf] (3) {3}}
        }
      };
      \draw[highlightedgecolor={color1!20}] (0) -- (in1) -- (root);
      \draw[highlightedgecolor={color2!20}] (1) -- (in2) -- (in3) -- (in1);
      \draw[highlightedgecolor={color3!20}] (2) -- (in2);
      \draw[highlightedgecolor={color4!20}] (3) -- (in3);
      \node[fit=(current bounding box),inner ysep=0.2cm,inner xsep=0] {};
    \end{tikzpicture}
    &
    \begin{tikzpicture}[scale=0.6,
      level 1/.style={level distance=1cm},
    ]
      \node (root) {$\rho$} [grow=down]
      child {node[inner,label=above right:{-3}] (in3) {}
        child {node[inner,label={-1}] (in1) {}
          child {node[leaf] (0) {0}}
          child {node[inner,label={-2}] (in2) {}
            child {node[leaf] (1) {1}}
            child {node[leaf] (2) {2}}
          }
        }
        child {node[leaf] (3) {3}}
      };
      \draw[highlightedgecolor={color1!20}] (0) -- (in1) -- (in3) -- (root);
      \draw[highlightedgecolor={color2!20}] (1) -- (in2) -- (in1);
      \draw[highlightedgecolor={color3!20}] (2) -- (in2);
      \draw[highlightedgecolor={color4!20}] (3) -- (in3);
      \node[fit=(current bounding box),inner ysep=0.2cm,inner xsep=0] {};
    \end{tikzpicture}
    &
    \begin{tikzpicture}[scale=0.6,
      level 1/.style={level distance=1cm},
    ]
      \node (root) {$\rho$} [grow=down]
      child {node[inner,label=above right:{-1}] (in1) {}
        child {node[left=0.2cm,inner,label={-3}] (in3) {}
          child {node[leaf] (0) {0}}
          child {node[leaf] (3) {3}}
        }
        child {node[right=0.2cm,inner,label={-2}] (in2) {}
          child {node[leaf] (1) {1}}
          child {node[leaf] (2) {2}}
        }
      };
      \draw[highlightedgecolor={color1!20}] (0) -- (in3) -- (in1) -- (root);
      \draw[highlightedgecolor={color2!20}] (1) -- (in2) -- (in1);
      \draw[highlightedgecolor={color3!20}] (2) -- (in2);
      \draw[highlightedgecolor={color4!20}] (3) -- (in3);
      \node[fit=(current bounding box),inner ysep=0.2cm,inner xsep=0] {};
    \end{tikzpicture}
    &
    \begin{tikzpicture}[scale=0.6,
      level 1/.style={level distance=1cm},
    ]
      \node (root) {$\rho$} [grow=down]
      child {node[inner,label=above right:{-1}] (in1) {}
        child {node[leaf] (0) {0}}
        child {node[inner,label={-2}] (in2) {}
          child {node[inner,label={-3}] (in3) {}
            child {node[leaf] (1) {1}}
            child {node[leaf] (3) {3}}
          }
          child {node[leaf] (2) {2}}
        }
      };
      \draw[highlightedgecolor={color1!20}] (0) -- (in1) -- (root);
      \draw[highlightedgecolor={color2!20}] (1) -- (in3) -- (in2) -- (in1);
      \draw[highlightedgecolor={color3!20}] (2) -- (in2);
      \draw[highlightedgecolor={color4!20}] (3) -- (in3);
      \node[fit=(current bounding box),inner ysep=0.2cm,inner xsep=0] {};
    \end{tikzpicture}
    &
    \begin{tikzpicture}[scale=0.6,
      level 1/.style={level distance=1cm},
    ]
      \node (root) {$\rho$} [grow=down]
      child {node[inner,label=above right:{-1}] (in1) {}
        child {node[leaf] (0) {0}}
        child {node[inner,label={-2}] (in2) {}
          child {node[leaf] (1) {1}}
          child {node[inner,label={-3}] (in3) {}
            child {node[leaf] (2) {2}}
            child {node[leaf] (3) {3}}
          }
        }
      };
      \draw[highlightedgecolor={color1!20}] (0) -- (in1) -- (root);
      \draw[highlightedgecolor={color2!20}] (1) -- (in2) -- (in1);
      \draw[highlightedgecolor={color3!20}] (2) -- (in3) -- (in2);
      \draw[highlightedgecolor={color4!20}] (3) -- (in3);
      \node[fit=(current bounding box),inner ysep=0.2cm,inner xsep=0] {};
    \end{tikzpicture}
    \\
    {\tt (0,1,-2)}
    &
    {\tt (0,1,-1)}
    &
    {\tt (0,1,0)}
    &
    {\tt (0,1,1)}
    &
    {\tt (0,1,2)}
  \end{tabular}
  }
  \caption{All phylogenetic trees with four leaves.}
  \label{fig:all-ola-4}
\end{figure}

\end{document}